\documentclass[journal, headsepline]{IEEEtran}

\usepackage{amsmath,amsfonts,amssymb,amsthm}
\usepackage{mathrsfs}
\usepackage{comment,blkarray}
\usepackage{multirow,bigdelim}
\usepackage{cite}
\usepackage[ruled,commentsnumbered, vlined]{algorithm2e}
\usepackage{tikz}
\usetikzlibrary{arrows,automata}
\usepackage[latin1]{inputenc}
\usepackage{verbatim}
\usepackage{graphicx}
\usepackage{cases}
\usepackage{booktabs}
\usepackage{caption}												
\usepackage{subcaption}	
\usepackage{etoolbox}

\usepackage{enumitem}
\usepackage{mathtools}

\makeatletter

\newtheorem{thm}{Theorem}
\newtheorem{lemma}{Lemma}
\newtheorem{cor}{Corollary}

\newtheorem{prop}{Proposition}
\theoremstyle{definition}
\newtheorem{example}{Example}
\allowdisplaybreaks

\begin{document}

\title{Robust Multidimensional Chinese Remainder Theorem (MD-CRT) with Non-Diagonal Moduli and Multi-Stage Framework}

\author{Guangpu Guo and Xiang-Gen Xia, \IEEEmembership{Fellow}, \IEEEmembership{IEEE} 
       
\thanks{G. Guo and X.-G. Xia are with the Department of Electrical and Computer Engineering,
  University of Delaware, Newark, DE 19716, USA
  (e-mails: guangpu@udel.edu and xxia@ee.udel.edu).
  This work was supported by the National Science Foundation (NSF) under Grant CCF-2246917.
}
}

\maketitle

\begin{abstract}
    The Chinese remainder theorem (CRT) provides an efficient way to reconstruct an integer from its remainders modulo several integer moduli, and has been widely applied in signal processing and information theory. Its multidimensional extension (MD-CRT) generalizes this principle to integer vectors and integer matrix moduli, enabling reconstruction in multidimensional signal processing scenarios. However, since matrices are generally non-commutative, the multidimensional extension introduces new theoretical and algorithmic challenges.
    When all matrix moduli are diagonal, the system is equivalent to applying the one-dimensional CRT independently along each dimension.
    This work first investigates whether non-diagonal (non-separable) moduli offer fundamental advantages over traditional diagonal ones.
    We show that under the same determinant constraint, non-diagonal matrices do not increase the dynamic range but yield more balanced and better-conditioned sampling patterns. More importantly, they generate lattices with longer shortest vectors, leading to higher robustness to vector remainder errors, compared to diagonal ones. To further improve the robustness, we develop a multi-stage robust MD-CRT framework that improves the robustness level without reducing the dynamic range.
    Due to the multidimensional nature and modulo matrix forms, it is challenging and not straightforward to extend the existing one-dimensional multi-stage robust CRT. In this paper, we obtain a new condition for matrix moduli, which can be easily checked, such that a multi-stage robust MD-CRT can be implemented.
    Both theoretical analysis and simulation results demonstrate that the proposed multi-stage robust MD-CRT achieves stronger error tolerance and more reliable reconstruction under erroneous vector remainders than that of single-stage robust MD-CRT.
\end{abstract}

\begin{IEEEkeywords}
Multidimensional Chinese remainder theorem (MD-CRT), dynamic range, robustness, multi-stage, non-separable sampling, shortest vectors.
\end{IEEEkeywords}

\IEEEpeerreviewmaketitle

\section{Introduction}\label{s1}
The Chinese remainder theorem (CRT) \cite{crt1} is a classical result in number theory. It provides an efficient way to reconstruct a large integer from its remainders modulo several smaller positive integers, called integer moduli.
As a powerful tool for recovering large or high-resolution quantities from smaller or folded measurements, the CRT has been widely used in many areas, including signal processing, coding theory, and cryptography \cite{crt1, crt2, xia99, gangli1, fft2, wenchaoli, radeee, xiaoli, congling1, congling2, lugan2,Akhlaq2015,Campobello2012,Chessa2012}. 
However, the traditional CRT is highly sensitive to errors. A small error in one remainder may lead to a large reconstruction error, making it unreliable in noisy environments. To address this issue, the robust CRT was introduced in \cite{xia07,li09,wang10,xiao2014,xiao17}. Its main idea is to have some redundancy among the moduli, such as to have a common integer factor among all the moduli. Then, when each remainder error is bounded within a certain range, the reconstruction error of the unknown integer can also be bounded within the same range. Similar ideas have also been extended to several generalized forms of the CRT \cite{li16,xhs18,xiao23a} with the same goal of achieving reliable reconstruction under small remainder errors.

In many modern applications, the quantities of interest are multidimensional vectors rather than scalars. Problems in imaging, radar, and communications often involve data with several dimensions. Motivated by this need, the CRT has been extended to the multidimensional CRT (MD-CRT) \cite{PPV1,MD1,MD2,gMDCRT}, where the integer moduli are replaced by integer matrices, called matrix moduli, and the remainders become vectors, called vector remainders. The MD-CRT provides a way to reconstruct an integer vector from its vector remainders modulo several integer matrix moduli. Similar to the one-dimensional case, the robustness issue also arises in the multidimensional setting. To handle this issue, the robust MD-CRT was proposed in \cite{MD1,MD2}, following the same basic idea as the robust CRT, i.e., to have some redundancy among the matrix moduli.

When all the matrix moduli are diagonal, the MD-CRT is equivalent to applying the one-dimensional CRT independently to each dimension. This naturally raises an interesting question: can non-diagonal (non-separable) matrices offer any fundamental advantage over diagonal (separable) ones, and if so, in what sense?
Although some existing works have studied general matrix moduli~\cite{MD1,MD2,gMDCRT}, they mainly focused on algorithmic formulations and reconstruction conditions rather than comparing their intrinsic performance with diagonal ones. In particular, whether such non-diagonal matrices can bring inherent benefits has not been systematically investigated. 

In \cite{guo25}, we briefly discussed one potential benefit from a sampling perspective. When the matrix moduli are viewed as sampling matrices, non-diagonal matrices can yield more uniform sampling patterns across different dimensions, helping avoid ill-conditioning in practical implementations. Moreover, as a special case of the two-dimensional CRT, the complex CRT (C-CRT) for complex-valued signals has been successfully applied to self-reset analog-to-digital converters (SR-ADCs), achieving better performance than traditional SR-ADCs using two separate real-valued 1D CRTs \cite{xiaopingli}. These observations suggest that non-diagonal matrices may have deeper theoretical advantages beyond implementation convenience.

In this paper, we study this question systematically from two key aspects: the dynamic range and the robustness.  
These two aspects are fundamental to the MD-CRT, as they jointly determine
the reconstruction capability and reliability of the system.  
However, since matrices are generally non-commutative, analyzing the MD-CRT is significantly more challenging than the conventional 1D-CRT.

The dynamic range of the MD-CRT represents the total number of uniquely determinable vectors for a given set of matrix moduli, or equivalently, the volume of the uniquely determinable range. 
Mathematically, it equals to the absolute determinant value of the least common right multiple (lcrm) of all the matrix moduli.  
When the MD-CRT is used in multidimensional sub-Nyquist sampling, each matrix modulus acts as a sampling matrix, and its absolute determinant corresponds to the sampling rate.  
Because of hardware constraints, all sampling rates are limited, i.e., the absolute determinant values of all the matrix moduli are upper bounded.  
Under the same determinant constraint, non-diagonal matrices naturally form a much larger candidate set than that of diagonal ones.  
This raises the question: can non-diagonal matrices achieve a larger dynamic range under the same determinant constraint?

In this paper, we first characterize the maximum dynamic range of the MD-CRT under a fixed determinant constraint. We show that this maximum dynamic range can be achieved by certain diagonal matrices. 
However, the number of such diagonal matrices is limited. When used as sampling matrices, their sampling rates are concentrated in one dimension, leading to unbalanced and ill-conditioned sampling patterns. 
In contrast, by using our previously proposed construction of pairwise co-prime non-diagonal matrices~\cite{guo25,primematrix}, we can obtain a set of non-diagonal matrices that preserve the same number of moduli and the same determinant for each matrix while producing more balanced and uniform sampling patterns across all dimensions.

In addition to the dynamic range, another key performance measure of the MD-CRT is its robustness level. 
We use the vector remainder error bound denoted by $\tau$ to quantify robustness, which means that if every vector remainder error is bounded by $\tau$, then the reconstruction error of the recovered vector is also bounded by $\tau$.  
Therefore, a larger $\tau$ implies stronger robustness and higher noise tolerance.
Following~\cite{MD2}, the value of $\tau$ is obtained from a
sufficient condition for robust multidimensional reconstruction. This bound
depends on the length of the shortest nonzero vector of certain lattices
generated by pairwise greatest common left divisors (gclds) of the moduli. As a sufficient bound, it does not claim the true maximum error tolerance. In practice, recovery may still succeed when errors slightly exceed $\tau$. Our simulations indicate that the gap between the sufficient bound and the empirical threshold is small, so the vector remainder error bound $\tau$ from \cite{MD2} serves as a practical and reliable indicator for comparing robustness across different sets of matrix moduli.

Based on the above definition of robustness, we compare diagonal and non-diagonal moduli under a fixed determinant constraint.
To solely investigate the robustness effect, we adopt the standard left-factor model: given a set of pairwise co-prime matrix moduli $\Gamma_1,\dots,\Gamma_L$, we study $\mathbf{M}\Gamma_1,\dots,\mathbf{M}\Gamma_L$, for which the sufficient bound $\tau$ is determined by the shortest nonzero vector of the lattice generated by~$\mathbf{M}$.
The problem then becomes: among all integer matrices $\mathbf{M}$ with $|\det(\mathbf{M})| \leq p$, can a non-diagonal matrix $\mathbf{M}$ generate a lattice whose shortest vector is longer than that of any diagonal matrix $\mathbf{M}$?

To make the analysis tractable, we focus on the two-dimensional case under a determinant constraint $p$, where $p$ is a prime integer. We show that the best diagonal matrix achieves a shortest-vector length $\lfloor\sqrt p\rfloor$, while certain non-diagonal matrices achieve a strictly larger value. We also develop a simple search procedure to identify such non-diagonal matrices and verify the improvement through simulations for all prime determinants below $10^5$. These results demonstrate that, although non-diagonal matrices do not increase the dynamic range, they can indeed provide a clear robustness advantage in the MD-CRT over that of diagonal ones. 

This difference between dynamic range and robustness comes from their basic nature.
The dynamic range is determined by the determinant of the lcrm of matrix moduli, which reflects the algebraic structure of the integer matrix ring.
Although non-diagonal matrices provide a larger candidate set, the maximum dynamic range is still limited by this algebraic relation.
In contrast, the robustness depends on the geometry of the lattices generated by each matrix modulus.
It is a local geometric property related to the lattice shape.
With more possible shapes to choose from, non-diagonal moduli can form lattices with longer shortest vectors, and therefore achieve stronger robustness than that diagonal moduli can.

In the second part of this paper, we further improve the robustness of the MD-CRT. To do so, there are two possible approaches. One is to build a multi-stage framework as that for 1D-CRT proposed in \cite{xiao2014}, and the other is to trade dynamic range for stronger robustness as that proposed for 1D-CRT in \cite{xiao17}. Although the two approaches are similar to that for 1D-CRT in \cite{xiao2014,xiao17}, they become much more challenging in multidimensional settings due to the matrix forms of the moduli.

Based on the algorithm in~\cite{MD2}, referred to as the single-stage robust MD-CRT, we propose a multi-stage robust MD-CRT framework.
The main idea in \cite{xiao2014} is to divide all matrix moduli into groups, apply the single-stage algorithm to each group, and then treat the reconstructed results as new moduli and remainders for the next stage.
The key challenge and difference compared to that for 1D-CRT in \cite{xiao2014} is that, in the multidimensional case,
the robustly determinable range of the robust MD-CRT is not always consistent with the uniquely determinable range of the MD-CRT for the same set of matrix moduli.
In fact, it cannot even be represented as a fundamental parallelepiped (FPD) in general.
This inconsistency makes the multi-stage framework difficult to implement.

To address this problem, we analyze the robustly determinable range of the single-stage robust MD-CRT and derive a sufficient condition under which it becomes consistent with the MD-CRT.
The obtained sufficient condition is easy to check and thus can be easily utilized in designing matrix moduli as we shall see later.
Based on this condition, the matrix moduli can be partitioned into multiple groups, where all moduli within each group satisfy this condition, and multi-stage reconstruction is then performed across these groups.
Our proposed multi-stage framework can further improve robustness without reducing the dynamic range.
Specifically, some matrix moduli that have no robustness in the single-stage case become having robustness in the multi-stage version, and for those already having robustness, the vector remainder error bound~$\tau$ can be further increased.
Simulation results confirm these findings.

The rest of this paper is organized as follows. Section~\ref{s2} reviews the necessary preliminaries. Section~\ref{s3} analyzes the dynamic ranges of diagonal and non-diagonal matrices. Section~\ref{s4} investigates the robustness advantage of non-diagonal matrices. Section~\ref{s5} presents the proposed multi-stage robust MD-CRT framework. Simulation results are provided in Section~\ref{s6}, and conclusions are given in Section~\ref{s7}.

\section{Preliminaries}\label{s2}
In this section, we briefly review the basic concepts and properties of integer lattices and matrices that are essential for analyzing the MD-CRT. For additional background on integer matrices and related results, the reader may refer to~\cite{MD2,matrix,remainder}.

Throughout this paper, $\mathbb{Z}$ denotes the set of integers and $\mathbb{R}$ denotes the set of real numbers.  
All vectors (e.g., $\mathbf{n}$, $\mathbf{f}$, $\mathbf{r}$) are $D$-dimensional integer vectors, and all matrices (e.g., $\mathbf{M}$, $\mathbf{N}$, $\mathbf{P}$) are $D\times D$ integer matrices, unless stated otherwise.  
We write $\mathbf{I}$ for the $D\times D$ identity matrix and $\mathbf{0}$ for the zero matrix or vector.  
The determinant of a matrix $\mathbf{M}$ is denoted by $\det(\mathbf{M})$. 
The operator $\text{diag}(\cdot)$ represents a diagonal matrix, and $|\mathcal{S}|$ denotes the cardinality of a set $\mathcal{S}$. Unless otherwise stated, the robustness statements and theorems in this paper hold for \emph{any} vector norm~$\|\cdot\|$.
For numerical evaluations and when computing shortest vectors (e.g., $\lambda_{\mathcal{L}(\mathbf{M})}$), we take $\|\cdot\|$ to be the Euclidean norm~$\|\cdot\|_2$. The notation $\left\lfloor \,\cdot\, \right\rfloor$ denotes the floor operation, applied element-wise when the argument is a vector.

\begin{enumerate}[leftmargin=1.3em]
  \item \textbf{Lattice:}  
  Given a nonsingular integer matrix $\mathbf{M} \in \mathbb{Z}^{D\times D}$, the lattice generated by $\mathbf{M}$ is defined as
  \[
  \mathcal{L}(\mathbf{M}) = \{ \mathbf{M}\mathbf{n} \mid \mathbf{n} \in \mathbb{Z}^D \}.
  \]

  \item \textbf{Shortest vector problem (SVP):}  
  The minimum distance of $\mathcal{L}(\mathbf{M})$ is
  \[
  \lambda_{\mathcal{L}(\mathbf{M})} = \min_{\mathbf{v}\in \mathcal{L}(\mathbf{M})\setminus\{\mathbf{0}\}} \|\mathbf{v}\|,
  \]
  which corresponds to the length of the shortest nonzero lattice vector.

  \item \textbf{Fundamental parallelepiped (FPD):}  
  For a nonsingular integer matrix $\mathbf{M}\in\mathbb{Z}^{D\times D}$, its FPD is defined as 
  \[
  \mathcal{N}(\mathbf{M}) = \{\mathbf{k}\in\mathbb{Z}^D \mid \mathbf{k} = \mathbf{M}\mathbf{x},~\mathbf{x}\in[0,1)^D\},
  \]
  and the number of elements in $\mathcal{N}(\mathbf{M})$ equals $|\det(\mathbf{M})|$.

  \item \textbf{Division representation for integer vectors:}  
  Any integer vector $\mathbf{f}\in\mathbb{Z}^D$ can be uniquely expressed as
  \[
  \mathbf{f} = \mathbf{M}\mathbf{n} + \mathbf{r},\quad \mathbf{r}\in\mathcal{N}(\mathbf{M}),~\mathbf{n}\in\mathbb{Z}^D,
  \]
  which is denoted as $\mathbf{f} \equiv \mathbf{r} \mod \mathbf{M}$, where $\mathbf{r}$ is called the \emph{vector remainder} of $\mathbf{f}$ modulo $\mathbf{M}$. 

  \item \textbf{Unimodular matrix:}  
  A square integer matrix $\mathbf{U}$ is unimodular if $|\det(\mathbf{U})|=1$. 

  \item \textbf{Greatest common left divisor (gcld):} For two integer matrices $\mathbf{M}$ and $\mathbf{N}$, a nonsingular integer matrix $\mathbf{G}$ is a \emph{common left divisor} (cld) of $\mathbf{M}$ and $\mathbf{N}$ if $\mathbf{G}^{-1}\mathbf{M}$ and $\mathbf{G}^{-1}\mathbf{N}$ are both integer matrices. Moreover, if any other cld of $\mathbf{M}$ and $\mathbf{N}$ is also a left divisor of $\mathbf{G}$, then $\mathbf{G}$ is called their \emph{greatest common left divisor} (gcld).

  \item \textbf{Least common right multiple (lcrm):} For two integer matrices $\mathbf{M}$ and $\mathbf{N}$, a matrix $\mathbf{R}$ is a \emph{common right multiple} (crm) if there exist integer matrices $\mathbf{P}$ and $\mathbf{Q}$ such that $\mathbf{R}=\mathbf{M}\mathbf{P}=\mathbf{N}\mathbf{Q}$. If any other crm of $\mathbf{M}$ and $\mathbf{N}$ is also a right multiple of $\mathbf{R}$, then $\mathbf{R}$ is called their \emph{least common right multiple} (lcrm). It is worth noting that the lcrm is not unique: if $\mathbf{R}$ is an lcrm of $\mathbf{M}$ and $\mathbf{N}$, then $\mathbf{R}\mathbf{U}$ is also an lcrm for any unimodular matrix $\mathbf{U}$. However, the absolute determinant value of the lcrm, $|\det(\mathbf{R})|$, is unique.

  \item \textbf{Smith normal form (SNF):}  
  Any integer matrix $\mathbf{M}\in\mathbb{Z}^{D\times K}$ can be factorized as
  \[
  \mathbf{U}\mathbf{M}\mathbf{V} = \mathbf{\Lambda},
  \]
  where $\mathbf{U} \in \mathbb{Z}^{D\times D}$ and $\mathbf{V} \in \mathbb{Z}^{K\times K}$ are unimodular matrices, and $\mathbf{\Lambda}$ is diagonal with entries $\delta_1,\dots,\delta_{\min\{K,D\}}$ satisfying $\delta_i\mid\delta_{i+1}$, i.e., $\delta_i$ divides $\delta_{i+1}$.

  \item \textbf{Hermite normal form (HNF):}  
  Any integer matrix $\mathbf{M}\in\mathbb{Z}^{D\times D}$ can be uniquely decomposed as
  \[
  \mathbf{M} = \mathbf{H}\mathbf{U},
  \]
  where $\mathbf{U} \in \mathbb{Z}^{D\times D}$ is unimodular and $\mathbf{H} \in \mathbb{Z}^{D\times D}$ is a lower-triangular integer matrix satisfying $0 \le H_{i,j} < H_{i,i}$ for $j < i$.

  \item \textbf{Coprimality:}  
  Matrices $\mathbf{M}$ and $\mathbf{N}$ are said to be left co-prime if their gcld is a unimodular matrix. Equivalently, $\mathbf{M}$ and $\mathbf{N}$ are left co-prime if and only if the SNF of matrix $(\mathbf{M} \ \mathbf{N})$ equals $(\mathbf{I} \ \mathbf{0})$. Throughout this paper, ``co-prime'' always refers to the left co-prime case.
\end{enumerate}

\section{Dynamic Range Analysis of MD-CRT}\label{s3}

In this section, we introduce the multidimensional undersampling model, and for MD-CRT, analyze the maximum dynamic range under a fixed determinant constraint. We then compare the performance of diagonal and non-diagonal matrix moduli.  
Our analysis shows that although non-diagonal matrices do not have a strictly larger dynamic range, they yield more balanced and better-conditioned sampling patterns than diagonal ones.  

\subsection{Multidimensional Undersampling and MD-CRT}

Consider a $D$-dimensional harmonic signal
\begin{equation}\label{eq:md_signal}
x(\mathbf{t}) = a\, e^{j 2\pi \mathbf{f}^{\top} \mathbf{t}} + \omega(\mathbf{t}),
\quad \mathbf{t}\in\mathbb{R}^{D},
\end{equation}
where $a$ is an unknown amplitude, $\mathbf{f}\in\mathbb{Z}^{D}$ is a $D$-dimensional integer frequency vector, and $\omega(\mathbf{t})$ denotes additive noise.
It is to determine $\mathbf{f}=[N_1,N_2,\ldots,N_D]^{\top}$ from several undersampled $D$-dimensional signals of $x(\mathbf{t})$, where each $N_i$ is a positive integer and may be large.

Let $\mathbf{M}_1,\mathbf{M}_2,\ldots,\mathbf{M}_L$ be $L$ nonsingular integer matrices of size $D\times D$, called \emph{sampling matrices}.
Sampling with $\mathbf{M}_i$ gives
\begin{equation}\label{eq:md_sample}
x_i[\mathbf{n}]
= a\, e^{j 2\pi \mathbf{f}^{\top}\mathbf{M}_i^{-\top}\mathbf{n}}
+ \omega[\mathbf{M}_i^{-\top}\mathbf{n}],
\quad \mathbf{n}\in\mathbb{Z}^{D},
\end{equation}
where the sampling density of $\mathbf{M}_i$ equals $|\det(\mathbf{M}_i)|$, i.e., there are $|\det (\mathbf{M}_i)|$ many sampled points per unit spatial volume of $\mathbb{R}^{D}$.

Applying the multidimensional DFT (MD-DFT) \cite{ar} to $x_i[\mathbf{n}]$ over $\mathbf{n}\in\mathcal{N}(\mathbf{M}_i^{\top})$, yields
\begin{equation}\label{eq:mddft}
X_i(\mathbf{k})
= a\,|\det(\mathbf{M}_i)|\,
\delta(\mathbf{k}-\mathbf{r}_i)
+ \Omega_i(\mathbf{k}), \quad \mathbf{k}\in\mathcal{N}(\mathbf{M}_i),
\end{equation}
where $\mathbf{r}_i$ is the integer vector remainder of $\mathbf{f}$ modulo $\mathbf{M}_i$, and $\delta(\cdot)$ denotes the discrete delta function.
Hence, each sampler provides one vector remainder $\mathbf{r}_i\equiv \mathbf{f}\bmod\mathbf{M}_i$. The MD-CRT enables reconstruction of $\mathbf{f}$ from these vector remainders.

\begin{prop}[MD-CRT~\cite{MD1}]\label{prop:MD-CRT}
Let $\mathbf{M}_1,\mathbf{M}_2,\ldots,\mathbf{M}_L$ be arbitrary nonsingular integer matrices of size $D\times D$, and let $\mathbf{R}$ be any lcrm of them.
For any integer vector $\mathbf{f}\in\mathbb{Z}^{D}$, it can be uniquely determined from its vector remainders $\mathbf{r}_i\equiv\mathbf{f}\bmod\mathbf{M}_i$, $1\le i\le L$, if and only if $\mathbf{f}\in\mathcal{N}(\mathbf{R})$.
\end{prop}

From Proposition~\ref{prop:MD-CRT}, the set of all uniquely determinable integer vectors is $\mathcal{N}(\mathbf{R})$, and its size $|\det(\mathbf{R})|$ defines the \emph{dynamic range} of the system.
In practice, one wants to minimize the sampling densities $|\det(\mathbf{M}_i)|$ while maximizing the dynamic range $|\det(\mathbf{R})|$. 
We next study the maximum dynamic range under a fixed determinant constraint, i.e., an upper bound on the sampling rates.

\subsection{Dynamic Range Analysis}

In the undersampling model introduced in the previous subsection, each matrix modulus serves as a sampling matrix whose absolute determinant represents the sampling rate. In many practical systems, such as moving-target parameter estimation in SAR imaging \cite{gangli1}, hardware limitations impose an upper bound on these sampling rates. This bound directly limits the possible matrix moduli, leading to a limited dynamic range.

To better understand the problem, we first recall the one-dimensional case.
For the classical CRT, if all the integer moduli are upper bounded by $q$, i.e., all integer moduli are less than or equal to $q$, where $q$ is a positive integer, then the maximum dynamic range equals the least common multiple (lcm) of $1,2,\cdots,q$, denoted by $\mbox{lcm}(1,2,\cdots,q)$.
This range can be achieved by choosing all integers from $1$ to $q$ as moduli, but that is clearly inefficient. A smaller subset of pairwise co-prime integers can already reach the same dynamic range, as shown below.

\begin{lemma}\label{lm1}
  Let $\{q_1,q_2,\cdots ,q_L\}\subseteq \{1,2,\cdots,q\}$ be a set of pairwise co-prime positive integers that maximizes the product $\prod_{i=1}^{L} q_i$. Then, $\prod_{i=1}^{L} q_i$ is equal to $\mbox{lcm}(1,2,\cdots,q)$.
\end{lemma}

\begin{proof}
  Because $q_1,q_2,\cdots ,q_L$ are pairwise co-prime integers, $\prod_{i=1}^{L} q_i$ equals the lcm of them. Since $\{q_1,q_2,\cdots ,q_L\}\subseteq \{1,2,\cdots,q\}$, we have $\mbox{lcm}(1,2,\cdots,q) \geq \prod_{i=1}^{L} q_i$. We then just need to prove $\mbox{lcm}(1,2,\cdots,q) \leq \prod_{i=1}^{L} q_i$ by proving that any integer $p$ with $1 \leq p \leq q$ is a divisor of $\prod_{i=1}^{L} q_i$. 

  Assume that there exists a $p$ with $1 \leq p \leq q$ such that $p$ is not a divisor of $\prod_{i=1}^{L} q_i$. Since any integer can be decomposed as the product of prime integers, there must be a divisor $r^{n}$ of $p$ where $r$ is a prime integer and $n$ is a positive integer, such that $r^{n}$ is not a divisor of $\prod_{i=1}^{L} q_i$. 
  
  If $r^{n}$ is co-prime with any $q_i$ for $1 \leq i \leq L$, we can construct a new set of pairwise co-prime integers $\{q_1,q_2,\cdots ,q_L,r^{n}\}\subseteq \{1,2,\cdots,q\}$ such that the product $r^{n}\prod_{i=1}^{L} q_i$ is greater than $\prod_{i=1}^{L} q_i$, which contradicts the assumption. 
  
  If $r^{n}$ is not co-prime with all $q_i$ for $1 \leq i \leq L$, there exists exactly one $q_j$ such that $q_j$ is not co-prime with $r^{n}$ since all $q_i$ for $1 \leq i \leq L$ are pairwise co-prime. Then, $q_j$ has a divisor $r^{a}$ for $1 \leq a <n$, otherwise $r^n$ would be a divisor of $q_j$ that would contradict with what is assumed for $r^n$. We then construct a new set of pairwise co-prime integers $\{q_1,q_2,\cdots, q_j/r^{a}, \cdots, q_L,r^{n}\}\subset \{1,2,\cdots,q\}$ such that their product $r^{n-a}\prod_{i=1}^{L} q_i$ is greater than $\prod_{i=1}^{L} q_i$, which contradicts the assumption. 

  Therefore, any integer $p$ with $1 \leq p \leq q$ is a divisor of $\prod_{i=1}^{L} q_i$ and then $\mbox{lcm}(1,2,\cdots,q) = \prod_{i=1}^{L} q_i$.
\end{proof}

Now consider the multidimensional case where each matrix modulus
$\mathbf{M}$ satisfies $|\det(\mathbf{M})|\le q$ for a positive integer $q$. In this setting, the maximum possible dynamic range is determined by the lcrm of all such matrices. However, unlike in the one-dimensional case, the number of such matrices is huge especially when the dimension $D$ is large. As a result, it is hard to describe the maximum possible dynamic range in a direct way. 
Before presenting the main result, we first present two lemmas that will be used for the main result. 

\begin{lemma}[Lemma 1 in \cite{guo25}]\label{lm2}
  Let $m_1$ and $m_2$ be two non-zero integers with the greatest common divisor (gcd) $k$ for a positive
  integer $k$.
  For each $i$ with $1 \leq i \leq D$, we can obtain a new matrix $[k\mathbf{e}_i \quad \mathbf{0}]$ by performing elementary column transformations on
   matrix $[m_1\mathbf{e}_i \quad m_2\mathbf{e}_i]$, where $\mathbf{e}_i$ is
  the $D$-dimensional vector with the $i$-th component $1$ and
  the other components $0$.
\end{lemma}

\begin{lemma}[Lemma 2 in \cite{MD1}]\label{lm3}
  Let $\mathbf{M}_i$ for $1 \leq i \leq L$ be $L$ nonsingular and pairwise co-prime and commutative integer matrices, then their product $\mathbf{M}_1\mathbf{M}_2\cdots\mathbf{M}_L$ is an lcrm of them.
\end{lemma}

We then proceed to present the main result.

\begin{thm}\label{th:maxdynamic}
  Assume that the absolute determinant values of all $D \times D$ matrix moduli are upper bound by $q$. Let $\{q_1,q_2,\cdots ,q_L\}\subset \{1,2,\cdots,q\}$ be a set of pairwise co-prime integers that maximizes the product $\prod_{i=1}^{L} q_i$ among all such sets. Then, the maximum possible dynamic range of MD-CRT under the upper bound $q$ of matrix moduli is $(q_1 q_2\cdots q_L)^{D}$.  
\end{thm}

\begin{proof}
  We first prove that the dynamic range of MD-CRT under the upper bound $q$ of matrix moduli is less than or equal to $(q_1 q_2\cdots q_L)^{D}$ by proving the $D\times D$ diagonal matrix
  \begin{equation*}
      \mathbf{R}=\mbox{diag}(q_1 q_2\cdots q_L,\cdots,q_1 q_2\cdots q_L)
  \end{equation*}
  is a crm of all the matrix moduli. 

  Let $\mathbf{M}$ be an integer matrix with the absolute determinant value $p$ for $1 \leq p \leq q$. We then have 
  \begin{align*}
        \mathbf{M}^{-1}\mathbf{R} &= \pm (1/p) \ \mbox{adj}(\mathbf{M}) \ \mbox{diag}(q_1 q_2\cdots q_L,\cdots,q_1 q_2\cdots q_L)\\
        &=\pm \ \mbox{adj}(\mathbf{M}) \ \mbox{diag}(\frac{q_1 q_2\cdots q_L}{p},\cdots,\frac{q_1 q_2\cdots q_L}{p}).
  \end{align*}
  From Lemma \ref{lm1}, the matrix $\mbox{diag}(\frac{q_1 q_2\cdots q_L}{p},\cdots,\frac{q_1 q_2\cdots q_L}{p})$ is an integer matrix. Therefore, $\mathbf{M}^{-1}\mathbf{R}$ is an integer matrix and $\mathbf{R}$ is a right multiple of $\mathbf{M}$. Since $\mathbf{R}$ is a right multiple of all possible integer matrices that the absolute determinant values are less than or equal to $q$, $\mathbf{R}$ is a crm of all such matrices. For MD-CRT, the uniquely determinable range is an FPD of an lcrm of all the matrix moduli and the dynamic range is equal to the absolute determinant value of the lcrm of all the matrix moduli. As a result, the maximum dynamic range of MD-CRT under the upper bound $q$ of matrix moduli is less than or equal to $|\det(\mathbf{R})|$, i.e., $(q_1 q_2\cdots q_L)^{D}$.

  We then prove that there is a set of integer matrices whose absolute determinant values are less than or equal to $q$ such that the dynamic range of MD-CRT with these matrix moduli is $(q_1 q_2\cdots q_L)^{D}$ by proving $\mathbf{R}$ is an lcrm of these matrices. 

  Let $\mathbf{D}_{i,1}=\mbox{diag}(q_i,1,\cdots,1)$, $\mathbf{D}_{i,2}=\mbox{diag}(1,q_i,\cdots,1)$, $\cdots$, $\mathbf{D}_{i,D}=\mbox{diag}(1,1,\cdots,q_i)$ for $1 \leq i \leq L$. We can easily check that the Smith normal form of matrix $(\mathbf{D}_{i_1,j_1} \ \mathbf{D}_{i_2,j_2})$ is $(\mathbf{I} \ \mathbf{0})$ for $1 \leq i_1\neq i_2 \leq L$ and $1 \leq j_1 \neq j_2 \leq D$ according to Lemma \ref{lm2}. As a result, all matrices $\mathbf{D}_{i,j}$ for $1 \leq i \leq L$ and $1 \leq j \leq D$ are pairwise co-prime and commutative. According to Lemma \ref{lm3}, the product of these matrices is an lcrm of them, which is the matrix $\mathbf{R}$. As a result, if we choose $\mathbf{D}_{i,j}$ as matrix moduli, the dynamic range is equal to $(q_1 q_2\cdots q_L)^{D}$.

  Combining the above two parts, we have completed the proof.
\end{proof}

According to Lemma \ref{lm1}, the maximum possible dynamic range of MD-CRT under the upper bound $q$ of matrix moduli is $\big{(}\mbox{lcm}(1,2,\cdots,q)\big{)}^{D}$. This result is consistent with the classical CRT when $D = 1$. The proof of Theorem \ref{th:maxdynamic} also shows that diagonal matrices are already sufficient to achieve the maximum dynamic range. Although non-diagonal matrices provide a larger candidate set, they cannot further increase this dynamic range. Next, we explore whether non-diagonal moduli can still bring advantages in other aspects of the MD-CRT.

\subsection{Sampling Advantage of Non-Diagonal Matrix Moduli}

While diagonal matrices achieve the maximum possible dynamic range, their sampling rates are often very uneven across dimensions.
For example, in $\mathbf{D}_{i,j}$, the sampling rate along the $j$-th dimension is $q_i$, and it is $1$ on the others. Such an imbalance may cause poor numerical conditioning and practical hardware difficulty. A natural question is whether these diagonal sampling matrices can be replaced by non-diagonal ones that have more balanced per-dimension sampling rates, without increasing the total number of matrices or changing their determinants. 

From Lemma \ref{lm2}, we have the following straightforward necessary and sufficient condition that two diagonal integer matrices are co-prime.

\begin{cor}\label{cor1}
    Two nonsingular diagonal integer matrices $\mbox{diag}(a_1,\cdots,a_D)$ and $\mbox{diag}(b_1,\cdots,b_D)$ are co-prime if and only if $a_i$ and $b_i$ are co-prime integers for all $1 \leq i \leq D$.
\end{cor}

Since $q_i$ for $1 \leq i \leq L$ are pairwise co-prime integers, to keep the maximum dynamic range, we must have $L$ groups of matrix moduli $\mathcal{S}_i$ for $1 \leq i \leq L$, such that $\mbox{diag}(q_i,\cdots,q_i)$ is an lcrm of the matrices in $\mathcal{S}_i$. 
Therefore, to analyze the replacement of diagonal matrices $\mathbf{D}_{i,j}$ with non-diagonal ones that achieve more balanced per-dimension sampling rates,
it suffices to focus on a single group corresponding to a fixed $q_i$. The analysis for other groups is similar.

Let $q_i=p_1^{n_1} p_2^{n_2}\cdots p_k^{n_k}$ be its prime factorization, where $k, n_1, n_2, \cdots, n_k$ are positive integers and $p_1,\cdots,p_k$ are prime integers.

Consider the special case $n_j=D$ for all $1\leq j \leq k$. In this case, $q_i$ can be written as $q_i=p_1^{D} p_2^{D}\cdots p_k^{D}$.
To construct $D$-dimensional matrices from these $k$ factors, 
we redefine a new set of $D$ pairwise co-prime integers 
$s_1,s_2,\ldots,s_D$ as follows:
if $k<D$, we append additional factors equal to $1$; 
if $k>D$, we combine some adjacent factors into a single one.
After this adjustment, $q_i$ can still be expressed as 
$q_i=s_1^{D}s_2^{D}\cdots s_D^{D}$, 
where $s_1,\ldots,s_D$ are pairwise co-prime integers.
One can then construct $D$ many pairwise co-prime 
diagonal matrices $\mathrm{diag}(s_1^{D},s_2^{D},\cdots,s_D^{D})$, $\mathrm{diag}(s_2^{D},s_3^{D},\cdots,s_1^{D})$, $\cdots$, $\mathrm{diag}(s_D^{D},s_1^{D},\cdots,s_{D-1}^{D})$.
For this diagonal matrices, the largest sampling rate along each dimension is $\max_{1 \leq j \leq D}s_j^{D}$. 
In this case, you cannot find another $D$ pairwise co-prime diagonal matrices with determinant $q_i$ such that the largest sampling rate along each dimension is less than $\max_{1 \leq j \leq D}s_j^{D}$.
However, following the construction in \cite{guo25}, we can have a set of $D$ non-diagonal matrices with determinant values $q_i$, whose lcrm is still $\mathrm{diag}(q_i,\ldots,q_i)$.
In this case, the maximum sampling rate along each dimension is $s_1s_2\cdots s_D+1$, which is much smaller than $\max_{1 \leq j \leq D}s_j^{D}$.
This shows that the non-diagonal matrices gives a more balanced sampling pattern compared with the diagonal construction.  
The special case $k=1$ has been studied in detail in \cite{guo25}.

For other cases of the exponents $n_j$, similar non-diagonal constructions can still be obtained by following the same approach in \cite{guo25}. They also lead to more balanced per-dimension sampling rates compared with the diagonal ones. We can also show that the resulting matrices are pairwise co-prime. 
Although our simulations indicate that their lcrm still equals $\mathrm{diag}(q_i,\ldots,q_i)$ for dimension $D\leq 75$, a complete theoretical proof remains open and will be left for future work.  

\section{Single-Stage Robust MD-CRT}\label{s4}

In low-SNR scenarios, detected vector remainders are often erroneous. The MD-CRT is highly sensitive to such errors, i.e., even a small error in a single vector remainder may cause a large error in the reconstructed vector. As a result, directly applying the MD-CRT under noise performs poorly. To address this issue, a robust MD-CRT for general matrix moduli was proposed in \cite{MD2}. It guarantees that if each vector remainder error is bounded by a moduli-dependent threshold, then the reconstruction error of the unknown integer vector is bounded by the same threshold. A larger threshold means a higher robustness level.
In this section, we briefly review the robust MD-CRT in \cite{MD2} and then analyze the potential robustness advantage of non-diagonal matrix moduli over diagonal matrix moduli, i.e., advantage of non-separable over separable robust MD-CRT, in the single-stage setting.  

\subsection{Single-Stage Robust MD-CRT}

Consider the following system of congruences:
\begin{equation}\label{cong}
    \begin{aligned}
      \mathbf{f} &= \mathbf{M}_1\mathbf{n}_{1}+\mathbf{r}_{1}, \\
      \mathbf{f} &= \mathbf{M}_2\mathbf{n}_{2}+\mathbf{r}_{2}, \\
       &\vdotswithin{=}\\
      \mathbf{f} &= \mathbf{M}_L\mathbf{n}_{L}+\mathbf{r}_{L},
    \end{aligned}
\end{equation}
where $\mathbf{f}$ is an unknown integer vector to be determined, $\mathbf{M}_i$, for $1 \leq i \leq L$, are nonsingular integer matrices and $\mathbf{r}_i \in \mathcal{N}(\mathbf{M}_i)$, for $1 \leq i \leq L$, are vector remainders of $\mathbf{f}$ modulo $\mathbf{M}_i$.
If all vector remainders $\mathbf{r}_i$ for $1 \leq i \leq L$ are error-free, then the unknown vector $\mathbf{f}$ can be uniquely determined by the MD-CRT if $\mathbf{f}$ lies within the FPD generated by an lcrm of the matrix moduli. 

In practical applications, the observed vector remainders are often noisy, i.e.,
\begin{equation*}
    \tilde{\mathbf{r}}_i=\mathbf{r}_i+\Delta\mathbf{r}_i, \quad 1 \leq i \leq L,
\end{equation*}
where ${\mathbf{r}_i} \in \mathcal{N}(\mathbf{M}_i)$ is an error-free vector remainder and $\Delta\mathbf{r}_i \in \mathbb{Z}^{D}$ represents an error.
If the MD-CRT is directly applied to these erroneous vector remainders $\tilde{\mathbf{r}}_i$, the resulting vector $\tilde{\mathbf{f}}$ constructed from MD-CRT may have a large error compared to the true vector $\mathbf{f}$, even when the errors $\Delta\mathbf{r}_i$ are small.

To overcome this issue, a robust MD-CRT for general matrix moduli has been proposed in \cite{MD2}. The key idea is to first accurately estimate the terms $\mathbf{M}_i\mathbf{n}_i$ for all $1 \leq i \leq L$ in the system of congruences (\ref{cong}) and then reconstruct the unknown vector $\mathbf{f}$ by averaging as follows:
\begin{equation*}
    \tilde{\mathbf{f}}=\frac{1}{L}\sum_{i=1}^{L}(\mathbf{M}_i\mathbf{n}_i+\tilde{\mathbf{r}}_i).
\end{equation*} 
In this case, the reconstruction error can be expressed as 
\begin{equation*}
    \tilde{\mathbf{f}} - \mathbf{f} =\frac{1}{L}\sum_{i=1}^{L} \Delta\mathbf{r}_i.
\end{equation*}
If the errors in the vector remainders are bounded by a parameter $\tau$, i.e., $\|\Delta\mathbf{r}_i\| \leq \tau$, for all $1 \leq i \leq L$, then the reconstruction error is also bounded by the same parameter:
\begin{equation*}
    \|\tilde{\mathbf{f}}-\mathbf{f} \| \leq \tau.
\end{equation*}
The parameter $\tau$ is called the vector remainder error bound, which is determined by the length of the shortest nonzero vector of certain lattices
generated by pairwise gclds of the moduli. It also reflects the level of robustness: a larger $\tau$ indicates a stronger tolerance to vector remainder errors and thus a higher robustness level.

Let $\mathbf{G}_{i,j}$ be a gcld of $\mathbf{M}_i$ and $\mathbf{M}_j$, for all $1 \leq i\neq j \leq L$. A necessary and sufficient condition for accurately estimating $\{\mathbf{M}_i\mathbf{n}_i\}_{i=1}^L$ in \cite{MD2} is given as follows.

\begin{prop} [Robust MD-CRT from \cite{MD2}] \label{robust}
Let  $\mathbf{M}_i$ for all $1 \leq i \leq L$ be $L$ distinct, arbitrary, nonsingular integer matrices.   
Suppose that there exists an index $l_0 \in \{1, 2, \dots, L\}$ satisfying
\begin{equation}\label{l0}
\min_{\substack{1 \leq j \leq L\\ j \neq l_0}} \lambda_{\mathcal{L}(\mathbf{G}_{l_0,j})}
= \max_{1 \leq i \leq L} \min_{\substack{1 \leq j \leq L\\ j \neq i}} \lambda_{\mathcal{L}(\mathbf{G}_{i,j})}.
\end{equation}

For any integer vector $\mathbf{f}$ with
\begin{equation}\label{dynamicrange}
\left\lfloor \mathbf{M}_{l_0}^{-1} \mathbf{f} \right\rfloor \in \mathcal{N} \left( \mathbf{M}_{l_0}^{-1} \text{lcrm}(\mathbf{M}_1, \mathbf{M}_2, \cdots, \mathbf{M}_L) \right),
\end{equation}
$\{\mathbf{M}_i\mathbf{n}_i\}_{i=1}^L$ can be accurately determined from the erroneous remainders $\{\tilde{\mathbf{r}}_i\}_{i=1}^L$, if and only if
\begin{equation}\label{iff}
\mathbf{0} = \arg\min_{\mathbf{h} \in \mathcal{L}(\mathbf{G}_{l_0,j})} \| \mathbf{h} - (\Delta \mathbf{r}_j - \Delta \mathbf{r}_{l_0}) \|,
\end{equation}
for all $1 \leq j \leq L$ with $j \neq l_0$.
Moreover, if the remainder errors are bounded by $\| \Delta \mathbf{r}_i \| \leq \tau$ for all  $1 \leq i \leq L$, then a sufficient condition for correct reconstruction of $\{\mathbf{M}_i\mathbf{n}_i\}_{i=1}^L$ is
\begin{equation}\label{tau}
\tau < \max_{1 \leq i \leq L} \min_{\substack{1 \leq j \leq L\\ j \neq i}} 
\frac{\lambda_{\mathcal{L}(\mathbf{G}_{i,j})}}{4}
= \min_{\substack{1 \leq j \leq L\\ j \neq l_0}} 
\frac{\lambda_{\mathcal{L}(\mathbf{G}_{l_0,j})}}{4}.
\end{equation}
\end{prop}

We now outline the procedure for robustly reconstructing $\mathbf{f}$ from the erroneous vector remainders $\{\tilde{\mathbf{r}}_i\}_{i=1}^L$ . 
\begin{description}[leftmargin=1.3cm]
  \item[Step 1:] Compute all $\mathbf{G}_{i,j}$ and their lattice minimum
  distances $\lambda_{\mathcal{L}(\mathbf{G}_{i,j})}$ for $1 \leq i\neq j \leq L$. Select an index $l_0$ that satisfies the condition in (\ref{l0}).

  \item[Step 2:] For each $j\neq l_0$, compute the vector $\mathbf{v}_j$ as 
  \begin{equation*}
      \mathbf{v}_j=\arg\min_{\mathbf{v}\in \mathcal{L}(\mathbf{G}_{l_0,j})} \|\mathbf{v}-(\tilde{\mathbf{r}}_j-\tilde{\mathbf{r}}_{l_0})\|.
  \end{equation*}

  \item[Step 3:] Compute $\mathbf{M}_{l_0}\mathbf{n}_{l_0} \in \mathcal{N}(lcrm(\mathbf{M}_1,\mathbf{M}_2,\cdots,\mathbf{M}_L))$ by applying the MD-CRT to solve the following system of congruences:
  \begin{equation}\notag
    \begin{aligned}
      \mathbf{M}_{l_0}\mathbf{n}_{l_0} &\equiv \mathbf{v}_{1} \mod{\mathbf{M}_{1}}, \\
      &\vdotswithin{\equiv}\\
      \mathbf{M}_{l_0}\mathbf{n}_{l_0} &\equiv \mathbf{v}_{l_0-1} \mod{\mathbf{M}_{l_0-1}}, \\
      \mathbf{M}_{l_0}\mathbf{n}_{l_0} &\equiv \mathbf{0} \mod{\mathbf{M}_{l_0}}, \\\mathbf{M}_{l_0}\mathbf{n}_{l_0} &\equiv \mathbf{v}_{l_0+1} \mod{\mathbf{M}_{l_0+1}}, \\
       &\vdotswithin{\equiv}\\
      \mathbf{M}_{l_0}\mathbf{n}_{l_0} &\equiv \mathbf{v}_{L} \mod{\mathbf{M}_{L}}.
    \end{aligned}
\end{equation}

   \item[Step 4:] For each $1 \leq j\leq L$ with $j\neq l_0$, Compute \begin{equation*}
       \mathbf{M}_j\mathbf{n}_j=\mathbf{M}_{l_0}\mathbf{n}_{l_0}-\mathbf{v}_j.
   \end{equation*}

   \item[Step 5:] Estimate $\mathbf{f}$ by averaging 
   \begin{equation*}
       \tilde{\mathbf{f}}=\frac{1}{L}\sum_{i=1}^{L}(\mathbf{M}_i\mathbf{n}_i+\tilde{\mathbf{r}}_i).
   \end{equation*}
\end{description}

The goal of the robust MD-CRT is to ensure that if the errors in the vector remainders are bounded by a parameter $\tau$, then the error in the reconstructed vector is also bounded by the same $\tau$.
Two questions arise in the robust MD-CRT.
The first concerns the dynamic range: how many integer vectors $\mathbf{f}$ can be robustly reconstructed? The second is about the robustness level: how large can the vector remainder error bound $\tau$ be that still guarantees robust reconstruction? 

For the first question, Proposition \ref{robust} shows that any integer vector $\mathbf{f}$ satisfying (\ref{dynamicrange}) can be robustly reconstructed if and only if condition \eqref{iff} is satisfied. Such an $\mathbf{f}$ can be uniquely decomposed as 
\begin{equation}\label{decomp}
    \mathbf{f}=\mathbf{M}_{l_0}\left\lfloor \mathbf{M}_{l_0}^{-1} \mathbf{f} \right\rfloor +\mathbf{r},
\end{equation}
where $\mathbf{r}\in \mathcal{N}(\mathbf{M}_{l_0})$. Hence, each pair $(\left\lfloor \mathbf{M}_{l_0}^{-1} \mathbf{f} \right\rfloor,\mathbf{r})$ uniquely determines an $\mathbf{f}$. According to (\ref{dynamicrange}), the number of possible vectors for $\left\lfloor \mathbf{M}_{l_0}^{-1} \mathbf{f} \right\rfloor$ is $$\frac{|\det(\text{lcrm}(\mathbf{M}_1,\mathbf{M}_2,\cdots,\mathbf{M}_L))|}{|\det(\mathbf{M}_{l_0})|},$$ and the total number of possible such vectors $\mathbf{r}$ is $|\det(\mathbf{M}_{l_0})|$. 
Therefore, the total number of vectors that can be robustly reconstructed is $|\det(\text{lcrm}(\mathbf{M}_1,\mathbf{M}_2,\cdots,\mathbf{M}_L))|$, which coincides with the dynamic range of the MD-CRT with the same set of matrix moduli. 
However, this does not mean that the set of vectors that can be robustly reconstructed with the robust MD-CRT, called \textit{robustly determinable range}, is the same as the set of vectors that can be uniquely determined by the MD-CRT, called \textit{uniquely determinable range}\footnote{For the conventional 1D-CRT, this uniquely determinable range is the same as the dynamic range. However, for the MD-CRT in multidimensional cases, it is the set of all uniquely determinable integer vectors from the vector remainders and different from the dynamic range.}. In fact, the robustly determinable range of the robust MD-CRT dose not, in general, form an FPD of any matrix, unlike that of MD-CRT, which is analyzed in detail in a later section. 

For the second question. According to Proposition \ref{robust}, a sufficient condition for robust reconstruction is given in (\ref{tau}).  
This condition shows that the robustness level is determined by the chosen matrix moduli. More Specifically, the robustness level is limited by the smallest minimum distance among all the lattices generated by the gclds of the chosen matrix with the other matrices. This motivates the question: can non-diagonal moduli yield a larger $\tau$ than diagonal ones under the same determinant constraint? We address this next.

\subsection{Robustness Advantage of Non-Diagonal Matrix Moduli}

In the previous section, comparing to diagonal matrix moduli we have shown that non-diagonal matrices do not increase the dynamic range of MD-CRT. For the robust MD-CRT, the dynamic range remains the same when the moduli are unchanged. The key performance index here is the vector remainder error bound $\tau$, which depends on the shortest vectors of certain lattices. We now ask whether non-diagonal matrices can yield a larger vector remainder error bound $\tau$ than diagonal matrices. 
Since solving the shortest vector problem (SVP) in lattices generated by integer matrices is computationally challenging in general, we only consider the two-dimensional case here to illustrate the potential advantages of non-diagonal matrices over diagonal matrices as moduli in the robust MD-CRT.

A direct comparison is difficult because~\eqref{tau} involves multiple
$\mathrm{gcld}$ lattices. To solely investigate the robustness, we adopt a left-factor model. Given a set of pairwise co-prime matrix moduli $\Gamma_1,\Gamma_2,\cdots,\Gamma_L$, consider new matrix moduli
$\mathbf{M}\Gamma_1,\mathbf{M}\Gamma_2,\cdots,\mathbf{M}\Gamma_L$. In this setting, the vector remainder error bound $\tau$ of the new matrix moduli equals $\lambda_{\mathcal{L}(\mathbf{M})}/4$ according to \eqref{tau}.

Since each modulus has a determinant constraint, $|\det(\mathbf{M})|$ is also upper bounded. 
Assume $|\det(\mathbf{M})|\le p$, where $p$ is a prime integer. The key question then becomes: does there exist a non-diagonal matrix $\mathbf{M}$ with $|\det(\mathbf{M})|\leq p$ such that the length of the shortest vector in the lattice generated by $\mathbf{M}$ is strictly larger than that of any diagonal matrix under the same constraint?   

We examine all $2\times2$ integer matrices of determinant $p$.
Let $$\mathbf{N}_{i}=\begin{pmatrix}
    1 &0\\
    i &p\\
\end{pmatrix},$$ for $i=0,1,2,\cdots,p-1$, where $p$ is a prime integer. 
Each $\mathbf{N}_i$ is in Hermite normal form (HNF), and together with $\mathrm{diag}(p,1)$ they represent all distinct HNFs of $2\times2$ integer matrices with determinant $p$. Matrices with the same HNF generate the same lattice. Hence all the possible lattices generated by all $2\times2$ integer matrices with determinant $p$ are exactly those generated by $\mathbf{N}_i$ for $0\le i<p$  and $\mathrm{diag}(p,1)$. 
The diagonal matrix $\mathrm{diag}(p,1)$ always generates a lattice whose shortest non-zero vector has length $1$, offering no robustness advantage. 
We therefore focus on $\mathbf{N}_i$ with $0 \leq i < p$.  

For the lattices generated by $\mathbf{N}_i$ with $0 \le i < p$, two natural questions arise: which lattice achieves the maximum shortest vector length, and what is the value of this maximum shortest vector length?  
In the following, we derive a general algorithm to compute $\max_{0 \leq i < p}\lambda_{\mathcal{L}(\mathbf{N}_i)}$ and the corresponding index $i$.  

The key idea is to find the smallest integer $d$ such that the set of integer vectors
\begin{equation}\label{eq:set(x,y)}
    \{[x,y]^{\top}| x^2+y^2\leq d\}\setminus \{[0,0]^{\top}\}
\end{equation} 
contains at least one vector from each lattice $\mathcal{L}(\mathbf{N}_i)$, $0 \leq i < p$. Then, $\sqrt{d}$ equals the
maximum shortest non-zero vector length over these lattices, i.e.,
\[
  \sqrt d \;=\; \max_{0\le i < p}\lambda_{\mathcal L(\mathbf N_i)}.
\] 

Since $[1,i]^{\top}\in\mathcal{L}(\mathbf{N}_i)$ for every 
$0\le i\le p-1$, the shortest nonzero vector length of 
$\mathcal{L}(\mathbf{N}_i)$ satisfies
\[
\lambda^2_{\mathcal L(\mathbf N_i)}\leq \|[1,i]^{\top}\|^2_2 
   =1+i^2 \leq 1+(p-1)^2 < p^2.
\]
Hence, it suffices to search for the maximum shortest vector length over 
$d < p^2$.
Moreover, since $[0,p]^{\top}$ is a basis vector of 
$\mathcal{L}(\mathbf{N}_i)$, the shortest nonzero lattice 
vector $[x,y]^{\top}$ with $x=0$ has length $p$, which is greater than $\lambda_{\mathcal L(\mathbf N_i)}$. If $[x,y]^{\top}
\in\mathcal{L}(\mathbf{N}_i)$, then so does 
$[-x,-y]^{\top}$. Therefore, it is sufficient to restrict 
the search to integer vectors $[x,y]^{\top}$ with 
$0 < x < p$ in the set~(\ref{eq:set(x,y)}).

For each lattice $\mathcal{L}(\mathbf{N}_i)$, every one of its lattice point has the form $[m,im+np]^{\top}$, where $m$ and $n$ are integers. Equivalently, $[x,y]^{\top}\in\mathcal{L}(\mathbf{N}_i)$ iff $xi\equiv y \mod{p}$.
When $p$ is prime, $\mathbb Z_p$ is a field under modulo-$p$ operations, so every nonzero $x\in\mathbb Z_p$ has a unique inverse $x^{-1}$. Hence, for any $[x,y]^{\top}$ with $0<x<p$ there exists a unique $i\in\mathbb Z_p$ such that $i\equiv yx^{-1} \mod p$ and $[x,y]^{\top}\in\mathcal{L}(\mathbf{N}_i)$.
This leads to the following algorithm to compute $\max_{i}\lambda_{\mathcal{L}(\mathbf{N}_i)}$ and the
corresponding maximizers $\arg\max_{i}\lambda_{\mathcal{L}(\mathbf{N}_i)}$.

\begin{description}[leftmargin=1.3cm]
  \item[Step 1:] Initialize $d=1$, index set
    $\mathcal S=\varnothing$. Enumerate all possible integer pairs $(x,y)$ such that $x^2+y^2=d$ and $0< x <p$.
  
  \item[Step 2:] For each pair $(x,y)$, calculate $i \equiv yx^{-1} \mod{p}$ and include $i$ in the set $\mathcal{S}_{\text{new}}$ for all the newly identified indices $i$ at the current step.  
  
  \item[Step 3:] Update the cumulative set $\mathcal{S}$ by $\mathcal{S}\cup\mathcal{S}_{\text{new}}$. If $|\mathcal{S}|=p$, i.e., $\mathcal{S}=\mathbb{Z}_p$, then the algorithm terminates, and the current distance $\sqrt{d}$ equals the maximum of the shortest vector length among all lattices, i.e., $\sqrt{d}=\max_{i}\lambda_{\mathcal{L}(\mathbf{N}_i)}$. The set $\mathcal{S}_{\text{new}}$ identified at this final step therefore corresponds to the indices of the lattices achieving this maximum shortest vector length, i.e., $\mathcal{S}_{\text{new}}=\arg\max_{i}\lambda_{\mathcal{L}(\mathbf{N}_i)}$. Otherwise, let $d=d+1$ and repeat the above steps.
\end{description}

Based on the above discussion, the procedure is guaranteed to terminate with $d< p^{2}$. A detailed algorithm is presented in Algorithm \ref{al1}.

\begin{algorithm}[htbp]
\caption{Finding the maximum shortest non-zero vector length among lattices $\mathcal{L}(\mathbf{N}_i)$, $0 \leq i < p$}
\label{al1}
\KwIn{A prime integer $p$}
\KwOut{The maximum shortest non-zero vector length $\sqrt d$ and a set of indices $i$  achieving it} 
Initialize $d \gets 1$, $\mathcal{S} \gets \emptyset$\;

\While{$|\mathcal{S}| < p$}{
    $\mathcal{S}_{\text{new}} \gets \emptyset$\;

    \For{$x \gets 1$ \KwTo $\lfloor \sqrt{d} \rfloor$}{
        $y^2 \gets d - x^2$\;

        \If{$y^2 < 0$ \textbf{or} $y^2$ is not a perfect square}{
            \textbf{continue}\;
        }

        $y \gets \sqrt{y^2}$\;

        \ForEach{$y' \in \{-y, +y\}$}{
            {
                $a \gets (y' \cdot x^{-1}) \bmod p$\;
                \If{$a \notin \mathcal{S}$}{
                    $\mathcal{S}_{\text{new}} \gets \mathcal{S}_{\text{new}} \cup \{a\}$\;
                }
            }
        }
    }

    $\mathcal{S} \gets \mathcal{S} \cup \mathcal{S}_{\text{new}}$\;

    \If{$|\mathcal{S}| = p$}{
        \Return{$\sqrt{d}, \mathcal{S}_{\text{new}}$}\;
    }

    $d \gets d + 1$\;
}
\end{algorithm}

For a diagonal matrix $\mathbf{D}=\mbox{diag}(a,b)$, the length of the shortest vector of $\mathcal{L}(\mathbf{D})$ is $\min\{|a|,|b|\}$. Since the determinant of $D$ is less than or equal to $p$, for a prime integer $p$, when $a=b=\lfloor \sqrt{p} \rfloor$, the length of the shortest vector achieves the maximum. In this case, the maximum possible length of the shortest vector of the lattice generated by diagonal matrices with determinant constraint $p$ is $\lfloor \sqrt{p} \rfloor$.

Based on Algorithm \ref{al1}, we compute the solutions for all prime integers $p<100{,}000$. In all cases, the maximum shortest vector length of the lattice generated by $\mathbf{N}_i$ for $1 \leq i < p$ is strictly greater than $ \lfloor\sqrt{p}\rfloor$. 

\begin{example}
Consider $p=3257$. For diagonal matrices, the maximum shortest vector length is $\lfloor \sqrt{p} \rfloor = 57$. In contrast, applying Algorithm~\ref{al1} shows that $\mathbf{N}_{971}$ yield lattices whose shortest vector length is $\sqrt{3730} \approx 61.07$.  
\end{example}

The above results illustrate that non-diagonal matrix moduli can provide a clear robustness advantage over diagonal ones.  
This implies that non-diagonal matrix moduli in robust MD-CRT, i.e., non-separable robust MD-CRT, outperforms the diagonal matrix moduli in robust MD-CRT, i.e., separable robust MD-CRT. To further improve the robustness of robust MD-CRT, we next turn to the second focus of this paper: how to further enhance the robustness of the robust MD-CRT through a multi-stage framework similar to the robust 1D-CRT studied in \cite{xiao2014}. As mentioned in Introduction, due to the matrix forms and the uniquely determinable range is different from the dynamic range in the multidimensional case, it is not a straightforward extension from the robust 1D-CRT in \cite{xiao2014} as we shall see later.  

\section{Multi-Stage Robust MD-CRT}\label{s5}

In the MD-CRT, in order to achieve a dynamic range as large as possible, we usually select as many pairwise co-prime integer matrix moduli as possible. However, if all the moduli are pairwise co-prime, then every gcld $\mathbf{G}_{i,j}$ of a pair matrix moduli $\mathbf{M}_i$ and $\mathbf{M}_j$ is a unimodular matrix and the corresponding lattice $\mathcal{L}(\mathbf{G}_{i,j})$ has only the minimum distance $\lambda_{\mathcal{L}(\mathbf{G}_{i,j})}=1$. From (\ref{tau}), this implies that the vector remainder error bound satisfies $\tau < 1/4$. Since vector remainders are integer vectors, no error can be tolerated in this case. 

Hence, to make the MD-CRT robust, some redundancy among the matrix moduli is necessary similar to that of the robust 1D-CRT.
According to Proposition \ref{robust}, there must be at least one matrix modulus that is not co-prime with any of the remaining matrix moduli. Otherwise, if for any matrix moduli there exists another matrix that is co-prime with it, then the error in the vector remainder cannot be corrected.

In practical applications, the sampling rates, i.e., the absolute determinant values of matrix moduli, are often constrained by hardware limits. Given a fixed set of matrix moduli, there are two possible strategies to improve robustness. The first approach is to reduce the dynamic range in exchange for a larger vector remainder error bound $\tau$, as what was studied in \cite{xiao17} for robust 1D-CRT. 
The second, which is the focus of this section, is to enhance robustness without reducing the dynamic range through a multi-stage reconstruction framework as what is studied for robust 1D-CRT in \cite{xiao2014}. 
The first approach remains an interesting open problem for future research.   

A key advantage of the multi-stage approach is that robust reconstruction can still be achieved even in the case where, for every matrix modulus, there exists another matrix modulus that is co-prime with it. This corresponds exactly to the scenario in which the single-stage robust MD-CRT fails as in Proposition~\ref{robust}. Moreover, in some cases where the original set of matrix moduli already allows robust reconstruction, the proposed multi-stage framework provides an even higher level of robustness, i.e., a larger vector remainder error bound $\tau$. We next present an example to illustrate the basic idea of the multi-stage framework.

\subsection{Motivating Example}
Let 
\begin{equation}\notag
\begin{aligned}
   \mathbf{\Gamma}_1 &=\begin{pmatrix}
       22 &-17\\
       17 &22\\
   \end{pmatrix}, 
   \mathbf{A}_1=\begin{pmatrix}
       16 &0\\
       1 &16\\
   \end{pmatrix}, 
   \mathbf{A}_2=\begin{pmatrix}
       16 &1\\
       0 &16\\
   \end{pmatrix},\\
   \mathbf{\Gamma}_2 &=\begin{pmatrix}
       22 &17\\
       -17 &22\\
   \end{pmatrix}, 
   \mathbf{B}_1=\begin{pmatrix}
       24 &0\\
       1 &24\\
   \end{pmatrix}, 
   \mathbf{B}_2=\begin{pmatrix}
       24 &1\\
       0 &24\\
   \end{pmatrix}.
\end{aligned}
\end{equation}
Consider the following six matrix moduli:
\begin{equation}\notag
\begin{aligned}
   \mathbf{M}_1 &= \mathbf{\Gamma}_1,\   
   \mathbf{M}_2=\mathbf{\Gamma}_1\mathbf{A}_1,\  
   \mathbf{M}_3=\mathbf{\Gamma}_1\mathbf{A}_2,\\
   \mathbf{M}_4 &= \mathbf{\Gamma}_2,\   
   \mathbf{M}_5=\mathbf{\Gamma}_2\mathbf{B}_1,\ 
   \mathbf{M}_6=\mathbf{\Gamma}_2\mathbf{B}_2.
\end{aligned}
\end{equation}
Since $\mathbf{\Gamma}_1$ and $\mathbf{\Gamma}_2$ are co-prime, we can easily check that for any matrix $\mathbf{M}_i$, $1 \leq i \leq 6$, there exists a matrix $\mathbf{M}_j$, $1 \leq j \neq i \leq 6$, such that they are co-prime. According to Proposition~\ref{robust}, the vector remainder error bound satisfies $\tau<1/4$, so the system is not robust. We can see that $\mathbf{M}_1,\mathbf{M}_2,\mathbf{M}_3$ share a common divisor $\mathbf{\Gamma}_1$ and $\mathbf{M}_4,\mathbf{M}_5,\mathbf{M}_6$ share a common divisor $\mathbf{\Gamma}_2$ and since $\mathbf{\Gamma}_1$ and $\mathbf{\Gamma}_2$ are co-prime, these 6 matrices do not have a common divisor that is non-unimodular.
This observation motivates dividing the moduli into two groups, $\{\mathbf{M}_1,\mathbf{M}_2,\mathbf{M}_3\}$, $\{\mathbf{M}_4,\mathbf{M}_5,\mathbf{M}_6\}$ and then apply the single-stage robust MD-CRT to each group separately. 

Before introducing how single-stage robust MD-CRT is applied, we first review how the MD-CRT operates when the matrix moduli are divided into several subgroups. 

Let the vector remainder of the unknown integer vector $\mathbf{f}$ modulo $\mathbf{M}_i$ be $\mathbf{r}_i$, for $1 \leq i \leq 6$. For the first group of matrix moduli, $\{\mathbf{M}_1,\mathbf{M}_2,\mathbf{M}_3\}$, if there exists an lcrm of $\mathbf{M}_1,\mathbf{M}_2,\mathbf{M}_3$, denoted by $\mathbf{R}_1$, such that the unknown vector $\mathbf{f}\in \mathcal{N}(\mathbf{R}_1)$, we can uniquely determine this vector $\mathbf{f}$. Otherwise, we can only obtain $\mathbf{f}_1\in \mathcal{N}(\mathbf{R}_1)$ that satisfies
\begin{equation}\label{r1}
    \mathbf{f}\equiv \mathbf{f}_1 \mod{\mathbf{R}_1}.
\end{equation}
Similarly, for the second group of matrix moduli, $\{\mathbf{M}_4,\mathbf{M}_5,\mathbf{M}_6\}$, if there exists an lcrm of $\mathbf{M}_4,\mathbf{M}_5,\mathbf{M}_6$, denoted by $\mathbf{R}_2$, such that the unknown vector $\mathbf{f}\in \mathcal{N}(\mathbf{R}_2)$, we can uniquely determine this vector $\mathbf{f}$. Otherwise, we can only obtain $\mathbf{f}_2\in \mathcal{N}(\mathbf{R}_2)$ that satisfies
\begin{equation}\label{r2}
    \mathbf{f}\equiv \mathbf{f}_2 \mod{\mathbf{R}_2}.
\end{equation}

Combining these two stages, the task becomes recovering $\mathbf{f}$ from a new set of matrix congruences (\ref{r1}) and (\ref{r2}). According to MD-CRT, $\mathbf{f}$ can be uniquely determined if and only if it is in the FPD of an lcrm of $\mathbf{R}_1$ and $\mathbf{R}_2$. If we instead apply the MD-CRT directly on all the six matrix moduli, $\mathbf{f}$ can be uniquely determined if and only if it is in the FPD of an lcrm of all the matrices $\mathbf{M}_i$ for $1 \leq i \leq 6$. Since $\mathbf{R}_1$ is an lcrm of $\mathbf{M}_1,\mathbf{M}_2,\mathbf{M}_3$ and $\mathbf{R}_2$ is an lcrm of $\mathbf{M}_4,\mathbf{M}_5,\mathbf{M}_6$, any lcrm of $\mathbf{R}_1$ and $\mathbf{R}_2$ is also an lcrm of $\mathbf{M}_i$ for $1 \leq i \leq 6$. Hence, in the error-free case, separately applying the MD-CRT to each subgroup produces the same reconstruction as applying it to the entire set of moduli.

Now, consider the case with erroneous vector remainders $\tilde{\mathbf{r}}_i$ for $1 \leq i \leq 6$ and denote the vector remainder error bound by $\tau$. For the first group of matrix moduli, $\{\mathbf{M}_1,\mathbf{M}_2,\mathbf{M}_3\}$, the gcld of any two matrices is $\mathbf{\Gamma}_1$. Therefore, condition~(\ref{l0}) in Proposition~\ref{robust} is satisfied for any choice of $l_0$ among the three matrices. Without loss of generality, let $l_0=1$. 
Then the unknown integer vector $\mathbf f$ must satisfy
\begin{equation}\label{dyn1}
\left\lfloor \mathbf{M}_1^{-1}\mathbf f \right\rfloor \in 
\mathcal N\!\left(\mathbf{M}_1^{-1}\mathbf R_1\right),
\end{equation}
where $\mathbf R_1$ is an lcrm of $\mathbf M_1,\mathbf M_2,\mathbf M_3$.  
If (\ref{dyn1}) holds, then by Proposition~\ref{robust}, 
robust reconstruction is guaranteed whenever $\tau < \lambda_{\mathcal{L}(\mathbf{\Gamma}_1)}/4 = \sqrt{773}/4$ using $\tilde{r}_i$, $i=1,2,3$.
In this case, we can compute an estimate $\tilde{\mathbf f}_1$ such that 
$\|\mathbf{f}-\tilde{\mathbf{f}}_1\| \leq \tau$.

For the second group of matrix moduli, $\{\mathbf{M}_4,\mathbf{M}_5,\mathbf{M}_6\}$, the gcld of any two matrices is $\mathbf{\Gamma}_2$.  
Therefore, condition~(\ref{l0}) in Proposition~\ref{robust} is satisfied for any choice of $l_0$ among the three matrices. Without loss of generality, let $l_0=4$. The unknown integer vector $\mathbf f$ must satisfy
\begin{equation}\label{dyn2}
\left\lfloor \mathbf{M}_4^{-1}\mathbf f \right\rfloor \in 
\mathcal N\!\left(\mathbf{M}_4^{-1}\mathbf R_2\right),
\end{equation}
where $\mathbf R_2$ is an lcrm of $\mathbf M_4,\mathbf M_5,\mathbf M_6$.  If (\ref{dyn2}) holds, then by Proposition~\ref{robust}, robust reconstruction is guaranteed when 
$\tau < \lambda_{\mathcal{L}(\mathbf{\Gamma}_2)}/4 = \sqrt{773}/4$ using $\tilde{r}_i$, $i=4,5,6$.  
In this case, we can compute an estimate $\tilde{\mathbf f}_2$ such that 
$\|\mathbf{f}-\tilde{\mathbf{f}}_2\| \leq \tau$.

In the second stage, the situation becomes even more complicated. Ideally, one would like to reduce the problem to the congruences 
$\mathbf f \equiv \mathbf f_1 \mod{\mathbf R_1}$ and 
$\mathbf f \equiv \mathbf f_2 \mod{\mathbf R_2}$ with erroneous vector remainders $\tilde{\mathbf{f}}_1,\tilde{\mathbf{f}}_2$.
If this was possible, then the reconstruction task in the second 
stage would fall exactly into the framework of the single-stage 
robust MD-CRT with two matrix moduli. However, to be able to apply the single-stage robust MD-CRT in the second stage,
an additional difficulty arises in the multidimensional robust setting as follows.

For the MD-CRT, given any unknown integer vector $\mathbf f$, 
even if $\mathbf f \notin \mathcal N(\mathbf R_1)$, 
we can always obtain a unique vector 
$\mathbf f_1 \in \mathcal N(\mathbf R_1)$ with the error-free vector remainders $\mathbf{r}_1,\mathbf{r}_2,\mathbf{r}_3$ such that 
$\mathbf f \equiv \mathbf f_1 \mod{\mathbf R_1}$.
Consequently, when the moduli are grouped, 
each subgroup produces an intermediate vector 
that can be treated as a new vector remainder 
for the next reconstruction stage  as explained earlier.

However, for the single-stage robust MD-CRT, this property no longer holds in general. 
The reason is that the robustly determinable range \eqref{dyn1} is usually \emph{not} an FPD of any integer matrix.  
In such cases, given the erroneous vector remainders $\tilde{\mathbf r}_1,\tilde{\mathbf r}_2,\tilde{\mathbf r}_3$, 
the reconstruction algorithm may output a vector 
$\tilde{\mathbf f}_1$ that has no corresponding 
$\mathbf f_1 \in \mathcal N(\mathbf R_1)$ 
satisfying both 
$\mathbf f \equiv \mathbf f_1 \mod{\mathbf R_1}$ 
and 
$\|\mathbf f_1-\tilde{\mathbf f}_1\|\le\tau$.
If $\tilde{\mathbf f}_1$ cannot be expressed 
as an erroneous version of the true vector remainder 
$\mathbf f_1 \in \mathcal N(\mathbf R_1)$, 
then it cannot serve as a valid vector remainder input for the next stage, 
because the congruence relation 
$\mathbf f \equiv \mathbf f_1 \mod{\mathbf R_1}$ 
itself is lost. As a result, we cannot directly cascade the single-stage robust MD-CRT into multiple stages.

In contrast, this issue never occurs 
in the one-dimensional case. 
When $D=1$, $M_1^{-1}R_1$ is an integer, 
so the robustly determinable range 
reduces to the set $\{0,1,\cdots,R_1-1\}$, 
which is the one-dimensional FPD $\mathcal{N}(R_1)$. 
Therefore, every integer $f$ admits a unique representative 
modulo $R_1$ even in the robust case, 
and the remainder estimated in the first stage 
can always serve as a valid input 
for the next reconstruction stage. 
It is the loss of FPD structure in higher dimensions that makes the multidimensional robust extension nontrivial and motivates the multi-stage framework developed in the following subsection.

\subsection{A Sufficient Condition for Multi-Stage Robustness}

As discussed in the previous subsection, a key requirement for the multi-stage framework, which applies to every pair of consecutive stages rather than only the first, is that the robustly determinable range at every intermediate stage (i.e., all stages except the final one) must form an FPD of some integer matrix.
Only in this case, the output of one stage can serve as a well-defined input for the next stage.
However, according to Proposition~\ref{robust}, the robustly determinable range in a single stage does not always have this structure.

We next characterize the structure of the robustly determinable range \eqref{dynamicrange} in Proposition~\ref{robust}.
As proposed above, any vector in the robustly determinable range has a unique decomposition given in (\ref{decomp}). This decomposition shows that such a vector $\mathbf{f}$ always lies in a shifted version of the FPD $\mathcal{N}(\mathbf{M}_{l_0})$. The shift is determined by the integer vector $\left\lfloor \mathbf{M}_{l_0}^{-1} \mathbf{f} \right\rfloor$, which belongs to the FPD $\mathcal{N} \left( \mathbf{M}_{l_0}^{-1} \text{lcrm}(\mathbf{M}_1, \mathbf{M}_2, \cdots, \mathbf{M}_L) \right)$.
More specifically, when $\left\lfloor \mathbf{M}_{l_0}^{-1} \mathbf{f} \right\rfloor=\mathbf{0}$, $\mathbf{f} \in \mathcal{N}(\mathbf{M}_{l_0})$. When $\left\lfloor \mathbf{M}_{l_0}^{-1} \mathbf{f} \right\rfloor=a\mathbf{e}_i$ for some integer $a$, the vector $\mathbf{f}$ lies in a shifted version of $\mathcal{N}(\mathbf{M}_{l_0})$ by $a\mathbf{m}_i$, where $\mathbf{m}_i$ is the $i$-th column of $\mathbf{M}_{l_0}$. That is, the FPD is shifted by $a\|\mathbf{m}_i\|_2$ along the direction of $\mathbf{m}_i$. 
More generally, when $\left\lfloor \mathbf{M}_{l_0}^{-1} \mathbf{f} \right\rfloor=\sum_{i=1}^{D}a_i\mathbf{e}_i$, the vector $\mathbf{f}$ lies in a shifted version of $\mathcal{N}(\mathbf{M}_{l_0})$ by $\sum_{i=1}^{D}a_i\mathbf{m}_i$. In this case, the shift occurs along each direction $\mathbf{m}_i$ with the distance $a_i\|\mathbf{m}_i\|_2$.

Since each shift moves the FPD by integer multiples of the length of a basis vector $\mathbf{m}_i$ in its direction for $1 \leq i \leq D$, the resulting shifted regions are disjoint. In other words, no two shifted FPDs overlap. Therefore, the entire region consists of $$\frac{|\det(\text{lcrm}(\mathbf{M}_1,\mathbf{M}_2,\cdots,\mathbf{M}_L))|}{|\det(\mathbf{M}_{l_0})|}$$ disjoint shifted versions of $\mathcal{N}(\mathbf{M}_{l_0})$, and the robustly determinable range in \eqref{dynamicrange} can be represented as 
\begin{equation}\label{union}
    \bigcup_{\mathbf{k} \in \mathcal{N}(\mathbf{M}_{l_0}^{-1}\text{lcrm}(\mathbf{M}_1,\mathbf{M}_2,\cdots,\mathbf{M}_L))} \big{(} \mathcal{N}(\mathbf{M}_{l_0}) + \mathbf{M}_{l_0}\mathbf{k} \big{)}.
\end{equation}
Note that different choices of $\text{lcrm}(\mathbf{M}_1,\mathbf{M}_2,\cdots,\mathbf{M}_L)$ in (\ref{union}), or \eqref{dynamicrange}, lead to different robustly determinable ranges.

We now give a simple example to show that the robustly determinable range in a single stage is not always an FPD. This motivates the following question: under what conditions is this union of disjoint shifted FPDs exactly a single FPD of some matrix?

\begin{example}\label{ex:shiftedFPD}
    Let 
    \begin{equation*}
        \mathbf{M}_1 =\begin{pmatrix}
        3 &1\\
        2 &2\\
        \end{pmatrix},
        \mathbf{M}_2 =\begin{pmatrix}
        2 &2\\
        1 &3\\
        \end{pmatrix}, 
    \end{equation*}
    we can calculate their lcrm \cite{MD2,PPV1} as 
    \begin{equation*}
        \mathbf{R}=\begin{pmatrix}
            4 &0\\
            0 &4\\
        \end{pmatrix}.
    \end{equation*}
    Then, we have $$\mathcal{N}(\mathbf{M}_1^{-1}\mathbf{R})=\{[0,0]^{\top},[1,0]^{\top},[0,1]^{\top},[1,-1]^{\top}\}.$$ The corresponding robustly determinable range is shown in Fig. \ref{fig:shiftedFPD}. This figure shows exactly four shifted copies of $\mathcal{N}(\mathbf{M}_1)$ and the dashed lines and hollow circles are not included in the robustly determinable range.
\begin{figure}[htbp]
    \centering
    \includegraphics[width=\columnwidth]{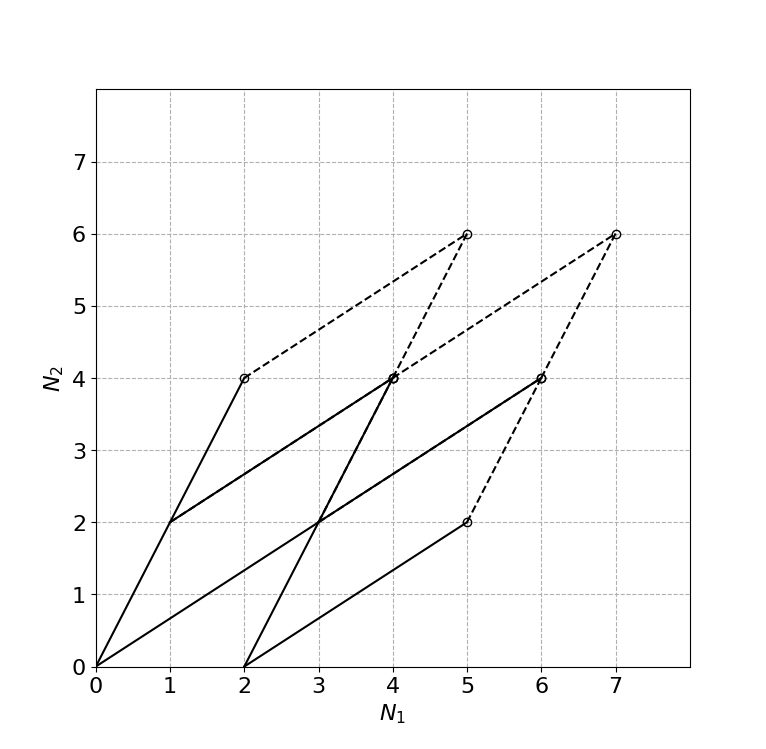} 
    \caption{The robustly determinable range in Example \ref{ex:shiftedFPD}}
    \label{fig:shiftedFPD}  
\end{figure}

Since both $[0,0]^{\top}$ and $[2,0]^{\top}$ lie in the robustly determinable range in Example \ref{ex:shiftedFPD} while $[1,0]^{\top}$ does not, this region cannot be represented as an FPD of any integer matrix.
\end{example}

This example shows that the robustly determinable range in a single-stage robust MD-CRT is not always an FPD of some matrix. Therefore, to construct a multi-stage framework in which each intermediate stage remains valid, we must select the matrix moduli and their lcrm carefully. The next lemma gives a sufficient condition that guarantees the region~\eqref{union} forms an FPD.

\begin{lemma}\label{lm:FPDregion}
    Let $\mathbf{M}_1,\cdots,\mathbf{M}_L$ be nonsingular integer matrices. If the Hermite normal form of the matrix $\mathbf{M}_{l_0}^{-1}\text{lcrm}(\mathbf{M}_1,\mathbf{M}_2,\cdots,\mathbf{M}_L)$ is a diagonal matrix, denoted by $\mathbf{D}$, then the robustly determinable range in (\ref{union}) can be expressed as the FPD of $\mathbf{M}_{l_0}\mathbf{D}$, when $\text{lcrm}(\mathbf{M}_1,\mathbf{M}_2,\cdots,\mathbf{M}_L)=\mathbf{M}_{l_0}\mathbf{D}$.
\end{lemma}

\begin{proof}
Since all lcrms of $\mathbf M_1,\ldots,\mathbf M_L$ share the same Hermite normal form, the Hermite normal form of $\mathbf M_{l_0}^{-1}\mathrm{lcrm}
(\mathbf M_1,\ldots,\mathbf M_L)$ is unique.
Let the Hermite normal form of the matrix $\mathbf{M}_{l_0}^{-1}\text{lcrm}(\mathbf{M}_1,\mathbf{M}_2,\cdots,\mathbf{M}_L)$ be $\mathbf{D}$, where $\mathbf{D}$ is a diagonal matrix. In other words, there exists a unimodular matrix $\mathbf{U}$ such that $\mathbf{M}_{l_0}^{-1}\text{lcrm}(\mathbf{M}_1,\mathbf{M}_2,\cdots,\mathbf{M}_L)\mathbf{U}=\mathbf{D}$. Then, $\mathbf{M}_{l_0}\mathbf{D}$ is also an lcrm of $\mathbf{M}_1,\cdots,\mathbf{M}_L$. When  $\mathbf{R}=\mathbf{M}_{l_0}\mathbf{D}$ is the chosen lcrm, then the corresponding robustly determinable range in \eqref{union} becomes
$$\mathcal{S}=\bigcup_{\mathbf{k} \in \mathcal{N}(\mathbf{D})} \big{(} \mathcal{N}(\mathbf{M}_{l_0}) + \mathbf{M}_{l_0}\mathbf{k} \big{)}.$$
We next prove that $\mathcal{S}=\mathcal{N}(\mathbf{R})$.

For any $\mathbf{s} \in \mathcal{S}$, there exist $\mathbf{x} \in [0,1)^{D}$  and $\mathbf{k} \in \mathcal{N}(\mathbf{D})$ such that $\mathbf{s}=\mathbf{M}_{l_0}(\mathbf{x}+\mathbf{k})$. Since $\mathbf{D}$ is a diagonal matrix that all diagonal elements are positive integers, WLOG, let $\mathbf{D}= \mbox{diag}(d_1,d_2,\cdots,d_D)$, where $d_1,d_2,\cdots,d_D$ are positive integers. Since $\mathbf{k} \in \mathcal{N}(\mathbf{D})$, $k_i \leq d_i-1$ for all $1 \leq i \leq D$. Therefore, $\mathbf{D}^{-1}(\mathbf{x}+\mathbf{k}) \in [0,1)^{D}$. Then, we have $\mathbf{s}=\mathbf{M}_{l_0}\mathbf{D}\big{(}\mathbf{D}^{-1}(\mathbf{x}+\mathbf{k})\big{)} \in \mathcal{N}(\mathbf{M}_{l_0}\mathbf{D})= \mathcal{N}(\mathbf{R})$. 

Conversely, for any $\mathbf{z} \in \mathcal{N}(\mathbf{R})$, there exists $\mathbf{y} \in [0,1)^{D}$ such that $\mathbf{z}=\mathbf{M}_{l_0}\mathbf{D}\mathbf{y}$. Write $\mathbf{D}\mathbf{y}= \lfloor \mathbf{D}\mathbf{y} \rfloor + \mathbf{r}$, with $\lfloor \mathbf{D}\mathbf{y} \rfloor \in \mathcal{N}(\mathbf{D})$ and $\mathbf{r} \in [0,1)^{D}$. Then we have $\mathbf{z}=\mathbf{M}_{l_0}\mathbf{r}+\mathbf{M}_{l_0} \lfloor \mathbf{D}\mathbf{y} \rfloor \in \mathcal{S}$.
\end{proof}

Lemma \ref{lm:FPDregion} provides a simple and effective condition for multi-stage reconstruction. It is easy to check in practice because the Hermite normal form of an integer matrix is unique and efficiently computable.
Moreover, when the condition is satisfied, the entire robustly determinable range coincides exactly with the FPD of a specific lcrm of the matrix moduli.
Thus, the robustly determinable range is the same with the uniquely determinable range of the MD-CRT.

Finally, note that this condition is only sufficient, not necessary.  
There exist cases where the robustly determinable range can still form an FPD even when the corresponding Hermite normal form is not diagonal.

\subsection{A Multi-Stage Robust Reconstruction Framework}
We now return to the motivating example introduced at the beginning of this section.

For the first group of matrix moduli, $\{\mathbf{M}_1,\mathbf{M}_2,\mathbf{M}_3\}$, if we select $l_0=1$, there exists an lcrm $\mathbf{R}_1$ of $\mathbf{M}_1,\mathbf{M}_2,\mathbf{M}_3$ such that 
\begin{equation*}
    \mathbf{R}_1=\mathbf{M}_1\begin{pmatrix}
        256 &0\\
        0 &256\\
    \end{pmatrix}.
\end{equation*}
According to Proposition \ref{robust} and Lemma \ref{lm:FPDregion}, if $\tau < \lambda_{\mathcal{L}(\mathbf{\Gamma}_1)}/4 = \sqrt{773}/4$, we can obtain an $\tilde{\mathbf{f}}_1$ with $
\|\mathbf{f}_1-\tilde{\mathbf{f}}_1\| \leq \tau$ by applying single-stage robust MD-CRT with matrix moduli $\mathbf{M}_1,\mathbf{M}_2,\mathbf{M}_3$ and their erroneous vector remainders $\tilde{\mathbf{r}}_1,\tilde{\mathbf{r}}_2,\tilde{\mathbf{r}}_3$, where $\mathbf{f}_1$ satisfies (\ref{r1}).

Similarly, for the second group of matrix moduli, $\{\mathbf{M}_4,\mathbf{M}_5,\mathbf{M}_6\}$, if we select $l_0=4$, there exists an lcrm $\mathbf{R}_2$ of $\mathbf{M}_4,\mathbf{M}_5,\mathbf{M}_6$ such that
\begin{equation*}
    \mathbf{R}_2=\mathbf{M}_4\begin{pmatrix}
        576 &0\\
        0 &576\\
    \end{pmatrix}.
\end{equation*}
According to Proposition \ref{robust} and Lemma \ref{lm:FPDregion}, if $\tau < \lambda_{\mathcal{L}(\mathbf{\Gamma}_2)}/4 = \sqrt{773}/4$, we can obtain an $\tilde{\mathbf{f}}_2 \in \mathcal{N}(\mathbf{R}_2)$ with $
\|\mathbf{f}_2-\tilde{\mathbf{f}}_2\| \leq \tau$ by applying single-stage robust MD-CRT with matrix moduli $\mathbf{M}_4,\mathbf{M}_5,\mathbf{M}_6$ and their erroneous vector remainders $\tilde{\mathbf{r}}_4,\tilde{\mathbf{r}}_5,\tilde{\mathbf{r}}_6$, where $\mathbf{f}_2$ satisfies (\ref{r2}).

Now we combine the two groups. Treat $\mathbf R_1,\mathbf R_2$ as new moduli and $\tilde{\mathbf f}_1,\tilde{\mathbf f}_2$ as their erroneous remainders.
Since $\mathbf M_1$ and $\mathbf M_4$ are co-prime, the gcld of $\mathbf{R}_1$ and $\mathbf{R}_2$ is $64\mathbf{I}$. Therefore, for any integer vector $\mathbf{f}$ with 
\begin{equation*}
\left\lfloor \mathbf{R}_{1}^{-1} \mathbf{f} \right\rfloor \in \mathcal{N} \left( \mathbf{R}_{1}^{-1} \text{lcrm}(\mathbf{R}_1, \mathbf{R}_2) \right),
\end{equation*}
if $\tau <\lambda_{\mathcal{L}(64\mathbf{I})}/4 = 16$, then a single-stage robust MD-CRT on $\{\mathbf R_1,\mathbf R_2\}$ yields an estimate
$\tilde{\mathbf f}$ with $\|\mathbf f-\tilde{\mathbf f}\|\le\tau$. 
Note that when $\mbox{lcrm}(\mathbf{R}_1,\mathbf{R}_2)
=\mbox{diag}(1780992,1780992)$, 
the robustly determinable range here is 
$\mathcal{N}(\mbox{lcrm}(\mathbf{R}_1,\mathbf{R}_2))$.

Combining these two stages, for any integer vector $\mathbf{f}$ with \begin{equation*}
\left\lfloor \mathbf{R}_{1}^{-1} \mathbf{f} \right\rfloor \in \mathcal{N} \left( \mathbf{R}_{1}^{-1} \text{lcrm}(\mathbf{R}_1, \mathbf{R}_2) \right),
\end{equation*}
if $\tau <\min\{\sqrt{773}/4,16\}=\sqrt{773}/4$, we can robustly estimate $\tilde{\mathbf{f}}$ with $
\|\mathbf{f}-\tilde{\mathbf{f}}\| \leq \tau$. In this example, if the robust MD-CRT is applied directly to all matrix moduli as a whole, no robustness can be achieved. However, by dividing the matrix moduli into two groups and then combining the results from each group to apply the robust MD-CRT again, robustness becomes attainable. This demonstrates that while the single-stage robust MD-CRT completely fails for this set of moduli, the proposed two-stage strategy enables non-trivial robustness.

The above example shows the key idea of multi-stage robust MD-CRT, we next formally present the framework. We first consider two-stage robust MD-CRT.

Let $\mathcal{M}$ be the full set of matrix moduli.  
For $1\le i\le K$, let
\[
\mathcal{M}_i=\{\mathbf{M}_{i,1},\mathbf{M}_{i,2},\ldots,\mathbf{M}_{i,L_i}\}
\]
be $K$ (not necessarily disjoint) subsets of $\mathcal{M}$ such that
$\bigcup_{i=1}^K \mathcal{M}_i \;=\; \mathcal{M}$, and $\mathbf{R}_i$ be an lcrm of the matrices in the $i$-th group $\mathcal{M}_i$. Each group $\mathcal M_i$ contains one or more matrices.  
If $L_i=1$, we directly set $\mathbf R_i=\mathbf M_{i,1}$.

For each group $\mathcal M_i$ with $L_i>1$, assume that there exists a matrix 
$\mathbf{M}_{i,1}\in \mathcal M_i$ such that the Hermite normal form of 
\[
\mathbf{M}_{i,1}^{-1}\,\mathrm{lcrm}(\mathbf{M}_{i,1},\mathbf{M}_{i,2},\ldots,\mathbf{M}_{i,L_i})
\]
is a diagonal matrix, denoted by $\mathbf{D}_i$.  
We then define $\mathbf{R}_i=\mathbf{M}_{i,1}\mathbf{D}_i$ as an lcrm of $\mathbf{M}_{i,1},\mathbf{M}_{i,2},\ldots,\mathbf{M}_{i,L_i}$, where $\mathbf{D}_i=\mathbf{I}$ if $L_i=1$, for $1\leq i\leq K$.  
Let $\mathbf {f}$ be an integer vector, and $\mathbf {r}_{i,j}$ denote the vector remainder of $\mathbf f$ modulo $\mathbf {M}_{i,j}$. Let $\mathbf{f}_i$ be the vector remainder of $\mathbf{f}$ modulo $\mathbf{R}_i$ for $1\leq i \leq K$.

Our objective is to reconstruct $\mathbf f$ through the following congruence system:
\begin{equation}\label{eq:systemall}
\mathbf{f}=\mathbf{M}_{i,j}\mathbf{n}_{i,j}+\mathbf{r}_{i,j}
\end{equation}
where $\mathbf{r}_{i,j}\in \mathcal{N}(\mathbf{M}_{i,j})$ and $\mathbf{n}_{i,j}$ is an unknown integer vector for $1 \leq i \leq K$ and $1 \leq j \leq L_i$. In the error-free case, $\mathbf{f}$ can be uniquely reconstructed from $\mathbf{M}_{i,j}$ and error-free vector remainders $\mathbf{r}_{i,j}$ for $1 \leq i \leq K$ and $1 \leq j \leq L_i$, if $\mathbf{f} \in \mathcal{N}(\text{lcrm}(\mathbf{R}_1,\mathbf{R}_2,\cdots,\mathbf{R}_K))$.

The system \eqref{eq:systemall} can be partitioned into $K$ subsystems:
\begin{equation}\label{group}
\mathbf{f}_i=\mathbf{M}_{i,j}\mathbf{n}^{\prime}_{i,j}+\mathbf{r}_{i,j}, \quad 1 \leq j \leq L_i,
\end{equation}
for $1 \leq i \leq K$.
Applying the MD-CRT to each subsystem \eqref{group} yields 
$\mathbf f_i\in \mathcal N(\mathbf R_i)$, which implies that $\mathbf{f}=\mathbf{R}_i\mathbf{n}_i+\mathbf{f}_i$, for some integer vector $\mathbf n_i$. 

These intermediate results $\{\mathbf f_i\}_{i=1}^K$ can then be combined to form a new system of congruences: 
\begin{equation}\label{secondstage}
   \mathbf{f} = \mathbf{R}_i\mathbf{n}_i+\mathbf{f}_i, \quad 1 \leq i \leq K,
\end{equation}
whose solution lies in $\mathcal N(\mathrm{lcrm}(\mathbf R_1,\mathbf R_2,\ldots,\mathbf R_K))$. If there exists a unimodular matrix $\mathbf U$ such that $\mathbf R_i=\mathbf R_j\mathbf U$ for some $i\neq j$, then one of the equations in \eqref{secondstage} is redundant and can be omitted.   
In the following, we assume that no such unimodular equivalence holds 
between any distinct pair $\mathbf R_i$ and $\mathbf R_j$.

We now turn to the case with erroneous vector remainders.  
Let $\tilde{\mathbf{r}}_{i,j}$ denote the erroneous vector remainder of $\mathbf f$ modulo $\mathbf M_{i,j}$ for $1 \leq i \leq K$ and $1 \leq j \leq L_i$.  
Based on these inputs, we present a robust reconstruction procedure for $\mathbf f$, referred to as the \emph{two-stage robust MD-CRT}.

\begin{description}[leftmargin=1.3cm]
  \item[Step 1:] For each subsystem \eqref{group} corresponding to $\mathcal M_i$:
  \begin{itemize}
     \item If $L_i>1$, fix $l_0=1$ and apply Steps 2--5 of the 
     single-stage robust MD-CRT in Section \ref{s4} to obtain a 
     robust estimate $\tilde{\mathbf f}_i$ of $\mathbf f_i$.  
     \item If $L_i=1$, simply set $\mathbf R_i=\mathbf M_{i,1}$ and 
     $\tilde{\mathbf f}_i=\tilde{\mathbf r}_{i,1}$.  
  \end{itemize}
  
  \item[Step 2:] Regard the estimates $\tilde{\mathbf f}_i$ as the 
  erroneous remainders in the system \eqref{secondstage}, and apply 
  Steps~1--5 of the single-stage robust MD-CRT in Section \ref{s4} to obtain the final 
  estimate $\tilde{\mathbf f}$ of $\mathbf f$.
\end{description}

The correctness and robustness of this two-stage procedure are established in the following theorem.

For each $1 \leq i \leq K$, let $\tau_{i}$ be the vector remainder error bound in the $i$-th group, i.e., $$\| \Delta \mathbf{r}_{i,j}\|= \| \tilde{\mathbf{r}}_{i,j} - \mathbf{r}_{i,j} \| \leq \tau_{i},$$ for $1 \leq j \leq L_i$. 

For each group $\mathcal{M}_i$ with $L_i>1$, let $\mathbf{G}^{(i)}_{j_1,j_2}$ be a gcld of $\mathbf{M}_{i,j_1}$ and $\mathbf{M}_{i,j_2}$ for $1\leq j_1 \neq j_2 \leq L_i$ and define 
\begin{equation*}
    \delta_i = \min_{\substack{2 \leq j \leq L_i}} 
\frac{\lambda_{\mathcal{L}(\mathbf{G}^{(i)}_{1,j})}}{4}.
\end{equation*}
For each group $\mathcal{M}_i$, the first matrix $\mathbf{M}_{i,1}$ 
is assumed to be the one satisfying Lemma~\ref{lm:FPDregion},
which plays the role of $\mathbf{M}_{l_0}$ in Proposition~\ref{robust}. 
Therefore, in computing vector remainder error bound, we fix $l_0=1$ for every group.
If $L_i=1$, we conventionally set $\delta_i=+\infty$ since a single modulus imposes no restriction at this stage.

Let $\mathbf{R}$ be an lcrm of $\mathbf R_1,\mathbf R_2,\ldots,\mathbf R_K$ and let $\mathbf{P}_{i,j}$ be a gcld of $\mathbf{R}_i$ and $\mathbf{R}_j$ for $1 \leq i \neq j \leq K$. Choose an index $l_0$ such that
\begin{equation}\notag
\min_{\substack{1 \leq j \leq K\\ j \neq l_0}} \lambda_{\mathcal{L}(\mathbf{P}_{l_0,j})}
= \max_{1 \leq i \leq K} \min_{\substack{1 \leq j \leq K\\ j \neq i}} \lambda_{\mathcal{L}(\mathbf{P}_{i,j})}.
\end{equation}
and define 
\begin{equation*}
    \delta = \min_{\substack{1 \leq j \leq K\\ j \neq l_0}} 
\frac{\lambda_{\mathcal{L}(\mathbf{P}_{l_0,j})}}{4}.
\end{equation*}

\begin{thm}\label{th:multi}
    Let $\mathbf{f}$ be an integer vector such that
    \begin{equation*}
        \mathbf{f} \in \bigcup_{\mathbf{k} \in \mathcal{N}(\mathbf{R}_{l_0}^{-1}\mathbf{R})} \big{(} \mathcal{N}(\mathbf{R}_{l_0}) + \mathbf{R}_{l_0}\mathbf{k} \big{)}.
    \end{equation*}
    If the vector remainder errors satisfy
    \begin{equation}\label{eq:th2-stage}
        \| \Delta \mathbf{r}_{i,j}\| \leq \tau_i < \min\{\delta_i,\delta\},
    \end{equation}
    for all $1 \leq i \leq K$ and $1 \leq j \leq L_i$,
    then $\mathbf f$ can be robustly reconstructed. In particular, the estimate $\tilde{\mathbf f}$ obtained by the two-stage robust MD-CRT satisfies
    \begin{equation}
        \| \tilde{\mathbf{f}}-\mathbf{f}\| \leq \frac{1}{K}\sum_{i=1}^{K}\tau_i.
    \end{equation}
\end{thm}

\begin{proof}
    For each $i$ with $L_i>1$, Lemma~\ref{lm:FPDregion} ensures that, with $\mathbf M_{i,1}$ chosen as $\mathbf M_{l_0}$ in Proposition~\ref{robust}, the robustly determinable range equals $\mathcal N(\mathbf R_i)$.

    Since $\mathbf{f}_i \in \mathcal{N}(\mathbf{R}_i)$ and $\| \Delta \mathbf{r}_{i,j}\| \leq \tau_i < \delta_i$, Proposition \ref{robust} ensures that $\mathbf{M}_{i,j}\mathbf{n}^{\prime}_{i,j}$ in \eqref{group} can be correctly determined. Thus, $\mathbf{f}_i$ can be robustly reconstructed as 
    \begin{equation*}
        \tilde{\mathbf{f}}_i= \frac{1}{L_i}\sum_{j=1}^{L_i}(\mathbf{M}_{i,j}\mathbf{n}^{\prime}_{i,j}+\tilde{\mathbf{r}}_{i,j}).
    \end{equation*}
    Moreover, we have 
    \begin{equation*}
       \| \tilde{\mathbf{f}}_i-\mathbf{f}_i\| \leq \| \frac{1}{L_i}\sum_{j=1}^{L_i}\Delta \mathbf{r}_{i,j}  \| \leq \tau_i.
    \end{equation*}

    If $L_i=1$, then $\mathbf R_i=\mathbf M_{i,1}$ and 
    $\tilde{\mathbf f}_i=\tilde{\mathbf r}_{i,1}$. In this case, the first stage reduces to a trivial identity and no further operation is needed.

    In the second stage, the estimates $\tilde{\mathbf{f}}_i$ serve as erroneous vector remainders in the system \eqref{secondstage}. Since 
    $$\mathbf{f} \in \bigcup_{\mathbf{k} \in \mathcal{N}(\mathbf{R}_{l_0}^{-1}\mathbf{R})} \big{(} \mathcal{N}(\mathbf{R}_{l_0}) + \mathbf{R}_{l_0}\mathbf{k} \big{)}$$
    and $\| \tilde{\mathbf{f}}_i-\mathbf{f}_i\| \leq \tau_i < \delta$, Proposition \ref{robust} guarantees that $\mathbf R_i\mathbf n_i$ can be accurately obtained. Therefore, $\mathbf f$ can be robustly reconstructed as
    \begin{equation*}
        \tilde{\mathbf{f}}= \frac{1}{K}\sum_{i=1}^{K}(\mathbf{R}_{i}\mathbf{n}_{i}+\tilde{\mathbf{f}}_{i}).
    \end{equation*}
    Finally, we obtain the error bound
     \begin{equation*}
        \| \tilde{\mathbf{f}}-\mathbf{f}\| \leq \frac{1}{K}\sum_{i=1}^{K}\tau_i.
    \end{equation*}
\end{proof}

The motivating example presented earlier has shown a case where the single-stage method cannot provide any robustness, while the two-stage framework successfully enables robust reconstruction.
We next present another example to show that even when the single-stage method already provides a certain level of robustness, the proposed multi-stage framework can further improve the robustness level.

\begin{example}\label{ex:2-stage}
    Let 
    \begin{equation}\notag
   \mathbf{C}_1 =\begin{pmatrix}
       7 &0\\
       1 &7\\
   \end{pmatrix}, 
   \mathbf{C}_2=\begin{pmatrix}
       7 &1\\
       0 &7\\
   \end{pmatrix}, 
   \mathbf{C}_3=\begin{pmatrix}
       11 &1\\
       0 &11\\
   \end{pmatrix}.
\end{equation}
These three matrices are pairwise co-prime. For any $2\times 2$ matrix $\mathbf{M}$, consider the following four matrix moduli:
$$30\mathbf{M},10\mathbf{M}\mathbf{C}_1,15\mathbf{M}\mathbf{C}_2,42\mathbf{M}\mathbf{C}_3.$$

For the single-stage robust MD-CRT, choosing $30\mathbf M$
as $\mathbf M_{l_0}$ achieves the maximum vector remainder error bound
$3\lambda_{\mathcal L(\mathbf M)}/2$.

For the two-stage robust MD-CRT, since $30\mathbf{M}$, $10\mathbf{M}\mathbf{C}_1,15\mathbf{M}\mathbf{C}_2$ satisfy Lemma \ref{lm:FPDregion}, we divide the matrix moduli into the following two groups:
$\{30\mathbf{M},10\mathbf{M}\mathbf{C}_1,15\mathbf{M}\mathbf{C}_2\}$ and $\{42\mathbf{M}\mathbf{C}_3\}$. For the first group, when $30\mathbf{M}$ is chosen as $\mathbf{M}_{l_0}$, the lcrm $\mathbf{R}_1$ is $1470\mathbf{M}$ and the robustness level $\delta_1=5\lambda_{\mathcal{L}_{(\mathbf{M}})}/2$. Since the second group contains only one matrix, it is directly carried into the second stage. So, the moduli in the second stage are $1470\mathbf{M}$ and $42\mathbf{M}\mathbf{C}_3$. The robustness level of the second stage is $\delta=21\lambda_{\mathcal{L}_{(\mathbf{M}})}/2$. As a result, the vector remainder error bound of the two-stage approach is $\min{\{\delta_1,\delta\}}=5\lambda_{\mathcal{L}_{(\mathbf{M}})}/2$, which is greater than that of the single-stage robust MD-CRT shown above.
\end{example}

In Theorem \ref{th:multi}, the vector remainder errors within the $i$-th group must satisfy $\|\Delta\mathbf r_{i,j}\|\le\tau_i<\delta_i$, so the vector remainder error bound is governed by $\delta_i$, which depends on the gclds $\{\mathbf G^{(i)}_{1,j}\}_{j=2}^{L_i}$. Hence, different partitions of the whole moduli set generally lead to different $\{\delta_i\}_{i=1}^K$. In the second stage, the vector remainder error bound is further constrained by $\delta$, which depends on the cross-group gclds $\{\mathbf P_{i,j}\}_{i\ne j}$.
Consequently, the overall robustness level of the two-stage procedure is $\tau_{\mathrm{overall}}=\min\{\delta_i,\delta\}$.
This shows that an appropriate grouping can increase the vector remainder error bound by enlarging some $\delta_i$ (and possibly $\delta$).

Unlike the one-dimensional case where groups can be formed arbitrarily,
in the multidimensional setting each group needs to satisfy the sufficient
condition in Lemma~\ref{lm:FPDregion} so that its robustly determinable
range is an FPD and the vector remainders in Stage-2 are well-defined. Therefore,
for a fixed set of matrix moduli, finding a grouping that maximizes
$\tau_{\mathrm{overall}}$ is nontrivial and is an interesting
direction for future work.

On the other hand, there are situations where multi-stage brings no
gain over the single-stage method. We next present such a case.

\begin{cor}
    For a given set of matrix moduli $\{\mathbf{M}_1,\mathbf{M}_2,\cdots,\mathbf{M}_L\}$, suppose $\mathbf{M}_i=\mathbf{M}\mathbf{\Gamma}_i$ for all $1 \leq i \leq L$, and $\mathbf{\Gamma}_i$ and $\mathbf{\Gamma}_j$ are co-prime for any $1 \leq i \neq j \leq L$. Then the vector remainder error bound cannot be improved by using a two-stage robust MD-CRT.
\end{cor}

\begin{proof}
Applying the single-stage robust MD-CRT directly to the whole set of matrix moduli yields the vector remainder error bound of $\lambda_{\mathcal{L}(\mathbf{M})}/4$.

Now suppose the moduli are divided into $K$ groups. For each group, since the $\Gamma_i$'s are pairwise co-prime, the vector remainder error bound of the single-stage robust MD-CRT within that group is also $\delta_i=\lambda_{\mathcal{L}(\mathbf{M})}/4$.
According to Theorem~\ref{th:multi}, the overall vector remainder error bound of the two-stage robust MD-CRT is $\min\{\delta_i,\delta\}$, which cannot exceed $\lambda_{\mathcal{L}(\mathbf{M})}/4$. Thus, no robustness gain can be achieved in this case.
\end{proof}

The two-stage procedure can be naturally extended to multiple stages.
Let $\mathcal{M}$ be the full set of matrix moduli, which is also the matrix moduli in the first stage. 
Let the full set of moduli be processed in $S$ stages.
At stage $s$, $1\le s\le S-1$, we form $K_s$ groups
\[
\mathcal{M}^{(s)}_i=\{\mathbf{M}^{(s)}_{i,1},\mathbf{M}^{(s)}_{i,2},\ldots,
\mathbf{M}^{(s)}_{i,L_i^{(s)}}\}, \quad 1\le i\le K_s,
\]
whose union is the full matrix moduli of stage $s$. 
If $L_i^{(s)}=1$, then we set $\mathbf R^{(s)}_i=\mathbf M^{(s)}_{i,1}$;  
otherwise, we assume there exists a matrix $\mathbf M^{(s)}_{i,1}\in 
\mathcal M^{(s)}_i$ such that the Hermite normal form of
\[
(\mathbf M^{(s)}_{i,1})^{-1}\,
\mathrm{lcrm}(\mathbf M^{(s)}_{i,1},\ldots,\mathbf M^{(s)}_{i,L_i^{(s)}})
\]
is a diagonal matrix $\mathbf D^{(s)}_i$.  
We then define $\mathbf R^{(s)}_i=\mathbf M^{(s)}_{i,1}\mathbf D^{(s)}_i$, 
which is an lcrm of all matrices in $\mathcal M^{(s)}_i$.  
Thus the output matrix moduli of stage $s$ is the collection 
$\{\mathbf R^{(s)}_i\}_{i=1}^{K_s}$, which becomes the input set of matrix moduli 
at stage $s+1$. Then, the matrix moduli of the final stage $S$ is $\mathbf R^{(S-1)}_1,\ldots,\mathbf R^{(S-1)}_{K_{S-1}}$.

For each group $\mathcal M^{(s)}_i$ with $L_i^{(s)}>1$, define
\begin{equation}\label{delta_i^s}
   \delta^{(s)}_i=\min_{2\le j\le L_i^{(s)}}
\frac{\lambda_{\mathcal L(\,\mathrm{gcld}(
\mathbf M^{(s)}_{i,1},\mathbf M^{(s)}_{i,j}))}}{4}. 
\end{equation}
If $L_i^{(s)}=1$, set $\delta^{(s)}_i:=+\infty$.
For the final stage $S$, let $\mathbf R^{(S)}$ denote an lcrm of 
$\mathbf R^{(S-1)}_1,\ldots,\mathbf R^{(S-1)}_{K_{S-1}}$, and define
\begin{equation}\label{delta}
   \delta=\min_{\substack{1 \leq j \leq K_{S-1} \\ j\ne l_0}}
\frac{\lambda_{\mathcal L(\,\mathrm{gcld}(
\mathbf R^{(S-1)}_{l_0},\mathbf R^{(S-1)}_j))}}{4}, 
\end{equation}
where $l_0$ is chosen as \eqref{l0} in Proposition~\ref{robust} for the set 
$\{\mathbf R^{(S-1)}_i\}_{i=1}^{K_{S-1}}$. 

We then define a path sequence for each initial group. 
For each group $\mathcal{M}_{i}^{(1)}$, $1 \leq i \leq K_1$, let
\[
\phi^{(1)}(i):=\{i\}, \qquad 
\phi^{(s)}(i)\subseteq \{1,\ldots,K_s\}, \quad 2 \leq s \leq S-1,
\]
where $\phi^{(s)}(i)$ denotes the index (or set of indices) of the group(s) 
at stage $s$ that receives the output of group $i$ at stage $1$.

The main result is as follows.  

\begin{thm}\label{thm:general-multistage}
Let $\mathbf f$ be an integer vector such that
\[
\mathbf f \in 
\bigcup_{\mathbf k\in\mathcal N((\mathbf R^{(S-1)}_{l_0})^{-1}\mathbf R^{(S)})}
\big(\mathcal N(\mathbf R^{(S-1)}_{l_0})+\mathbf R^{(S-1)}_{l_0}\mathbf k\big).
\]
If the vector remainder errors satisfy
\begin{equation}\label{eq:thk-stage}
   \|\Delta \mathbf r_{i,j}\|\ \le\ \tau_i\ 
<\ \min\Big\{ \delta^{(1)}_i,\ 
\min_{\substack{2 \leq s \leq S-1 \\ k \in \phi^{(s)}(i)}} \delta^{(s)}_k,\ 
\delta \Big\}, 
\end{equation}
for $
1\le i\le K_1,\;1\le j\le L_i^{(1)}$, then $\mathbf f$ can be robustly reconstructed.  
Moreover, the final estimate $\tilde{\mathbf f}$ satisfies
\[
\|\tilde{\mathbf f}-\mathbf f\|\ 
\le \frac{1}{K_1}\sum_{i=1}^{K_1}\tau_i.
\]
\end{thm}

The proof follows exactly the same argument as in the two-stage case, 
hence it is omitted.

To conclude this section, we present an example showing that a 
three-stage robust MD-CRT can achieve a higher vector remainder error bound 
than the two-stage framework.

\begin{example}\label{ex:three-vs-two}
   Let \begin{equation}\notag
   \mathbf{E}_{i1} =\begin{pmatrix}
       e_i &1\\
       0 &e_i\\
   \end{pmatrix}, 
   \mathbf{E}_{i2}=\begin{pmatrix}
       e_i &0\\
       1 &e_i\\
   \end{pmatrix},
\end{equation} 
for $1\leq i \leq 5$, where $e_1=8,e_2=16,e_3=36,e_4=9,e_5=27$. Let $\mathbf{\Gamma}_1$ and $\mathbf{\Gamma}_2$ be as in the motivating example and
\begin{equation}\notag
   \mathbf{\Gamma}_{3} =\begin{pmatrix}
       1 &0\\
       53 &769\\
   \end{pmatrix},
   \mathbf{\Gamma}_{4} =\begin{pmatrix}
       26 &-11\\
       11 &26\\
   \end{pmatrix},
   \mathbf{\Gamma}_{5} =\begin{pmatrix}
       26 &11\\
       -11 &26\\
   \end{pmatrix},
\end{equation} 
All $\mathbf{\Gamma}_j$ are pairwise co-prime and $\lambda_{\mathcal{L}(\mathbf{\Gamma_1})}=\lambda_{\mathcal{L}(\mathbf{\Gamma_2})}=\sqrt{773}$,$\lambda_{\mathcal{L}(\mathbf{\Gamma_3})}=\sqrt{842}$, $\lambda_{\mathcal{L}(\mathbf{\Gamma_4})}=\lambda_{\mathcal{L}(\mathbf{\Gamma_5})}=\sqrt{797}$.

Consider the following matrix moduli:
$$\{\mathbf{\Gamma}_k,\mathbf{\Gamma}_k\mathbf{E}_{k1},\mathbf{\Gamma}_k\mathbf{E}_{k2}\ |\ 1 \leq k \leq 5\},$$
Since all $\mathbf{\Gamma}_j$ are pairwise co-prime, the vector remainder error bound of single-stage robust MD-CRT is $1/4$, which means that no error can be tolerated.

We then divide them into five groups:
$$\mathcal{M}_{k}^{(1)}=\{\mathbf{\Gamma}_k,\mathbf{\Gamma}_k\mathbf{E}_{k1},\mathbf{\Gamma}_k\mathbf{E}_{k2}\}, \quad  \text{for} \ 1 \leq k \leq 5,$$
and let the remainder errors of each group $k$ be bounded by $\tau_k$.  
Each $\mathcal M_k^{(1)}$ satisfies Lemma~\ref{lm:FPDregion} with
$\mathbf M_{l_0}=\mathbf\Gamma_k$. Hence, by \eqref{delta_i^s}, $\delta_1^{(1)}=\delta_2^{(1)}=\sqrt{773}/4\approx6.95$, $\delta_3^{(1)}=\sqrt{842}/4\approx7.25$, $\delta_4^{(1)}=\delta_5^{(1)}=\sqrt{797}/4\approx7.05$. 
Moreover, the lcrms satisfying Lemma~\ref{lm:FPDregion} are
\begin{equation}\notag
   \mathbf{R}^{(1)}_{k} =\mathbf{\Gamma}_{k}\begin{pmatrix}
       r_k &0\\
       0 &r_k\\
   \end{pmatrix}, \quad 1\leq k\leq 5,
\end{equation} 
where $r_1=64,r_2=256,r_3=1296,r_4=81,r_5=729$. Then, all these $\mathbf{R}^{(1)}_{k}$ serve the inputs of the second stage. If we use the single-stage robust MD-CRT directly to the second stage, the vector remainder error bound is $\delta_{2-stage}=4$.

Instead, we split the matrix moduli in the second stage into two groups: $\mathcal{M}^{(2)}_1=\{\mathbf{R}^{(1)}_{1},\mathbf{R}^{(1)}_{2},\mathbf{R}^{(1)}_{3}\}$ and $\mathcal{M}^{(2)}_2=\{\mathbf{R}^{(1)}_{3},\mathbf{R}^{(1)}_{4},\mathbf{R}^{(1)}_{5}\}$. Both groups satisfy Lemma~\ref{lm:FPDregion} with
$\mathbf M_{l_0}=\mathbf R^{(1)}_3$. Then, we can calculate vector remainder error bound of each group by \eqref{delta_i^s} and obtain $\delta_1^{(2)}=4$, $\delta_2^{(2)}=81/4=20.25$. 
The resulting lcrms (still satisfying Lemma~\ref{lm:FPDregion}) are
\begin{equation}\notag
   \mathbf{R}^{(2)}_{k} =\mathbf{\Gamma}_{3}\begin{pmatrix}
       r^{\prime}_k &0\\
       0 &r^{\prime}_k\\
   \end{pmatrix}, \quad k=1,2,
\end{equation} 
where $r^{\prime}_1=16028928,r^{\prime}_2=9296208$. Then, all two $\mathbf{R}^{(2)}_{k}$ serve the inputs of the third stage. Use single-stage robust MD-CRT on $\{\mathbf R^{(2)}_1,\mathbf R^{(2)}_2\}$, the vector remainder error bound is $\delta_{3-stage}=324$.
According to the grouping structure, the flow of remainders can be
described as follows:
$\phi^{(2)}(1)=\{1\}$,
$\phi^{(2)}(2)=\{1\}$, 
$\phi^{(2)}(3)=\{1,2\}$,
$\phi^{(2)}(4)=\{2\}$,
$\phi^{(2)}(5)=\{2\}$.

According to \eqref{eq:thk-stage}, we can explicitly compute the 
upper bounds on the remainder error $\tau_k$ for each 
initial group $\mathcal{M}_{k}^{(1)}$. These values are summarized in 
Table~\ref{tab:robustness_comparison}. From the table, it is clear 
that while the single-stage approach allows no error, 
the two-stage framework improves the vector remainder error bound uniformly 
across all groups, and the three-stage framework further enhances 
the bounds for certain groups, specifically groups~4 and~5.

\begin{table}[htbp]
\centering
\caption{Vector remainder error bounds for different groups.}
\label{tab:robustness_comparison}
\begin{tabular}{c|ccccc}
\toprule
\text{Groups} & $\tau_1$ & $\tau_2$ & $\tau_3$ & $\tau_4$ & $\tau_5$ \\
\midrule
Single-stage & 0.25 & 0.25 & 0.25 & 0.25 & 0.25 \\
Two-stage    & 4    & 4    & 4    & 4    & 4    \\
Three-stage  & 4    & 4    & 4    & 7.05 & 7.05 \\
\bottomrule
\end{tabular}
\end{table}
\end{example}

\section{Simulations}\label{s6}

In this section, we present numerical simulations to verify the theoretical results on the robustness advantage of non-diagonal matrix moduli in robust MD-CRT, i.e., non-separable robust MD-CRT,
and the proposed multi-stage robust MD-CRT framework developed in Theorem~\ref{th:multi}.

We first consider the four matrix moduli in Example~\ref{ex:2-stage}, namely $30\mathbf{M}$, $10\mathbf{M}\mathbf{C}_1$, $15\mathbf{M}\mathbf{C}_2$, and $42\mathbf{M}\mathbf{C}_3$.
As discussed in Example~\ref{ex:2-stage}, the vector remainder error bound of the single-stage scheme is $3\lambda_{\mathcal{L}(\mathbf{M})}/2$, while that of the two-stage scheme is $5\lambda_{\mathcal{L}(\mathbf{M})}/2$.

Assume the determinant constraint of $\mathbf{M}$ is $881$. Under this constraint, one of the diagonal matrices achieving the maximum shortest vector length among diagonal integer matrices is 
\[
\mathbf{M}=\begin{pmatrix}
29 & 0\\
0 & 29
\end{pmatrix},
\]
with $\lambda_{\mathcal{L}(\mathbf{M})}=29$.
According to Algorithm~\ref{al1}, one of the non-diagonal matrices achieving the maximum shortest vector length among integer matrices with determinant $881$ is  
\[
\mathbf{M}^{\prime}=\begin{pmatrix}
1 & 0\\
32 & 881
\end{pmatrix},
\]
with $\lambda_{\mathcal{L}(\mathbf{M}^{\prime})}=\sqrt{1009}\approx31.76$.

For both matrices, a fixed integer vector $\mathbf{f}$ is chosen
within the robustly determinable range.  
For each vector remainder error bound $\tau=5,10,\ldots,85$, error vectors $\Delta\mathbf{r}_i$ satisfying
$\|\Delta\mathbf{r}_i\|_2\le\tau$ are generated uniformly and independently within
$\{\mathbf{e}\in\mathbb{Z}^2:\ \|\mathbf{e}\|_2\le\tau\}$ for $i=1,2,3,4$. 
Each case is averaged over $2000$ independent trials, and the mean reconstruction error  $\mathbb{E}[\|\mathbf{f}-\tilde{\mathbf{f}}\|_2]$ is plotted versus~$\tau$ in Fig.~\ref{fig:sim1}.
As shown in the figure, the single-stage system with the diagonal matrix
loses robustness when $\tau$ exceeds 45.
In contrast, the non-diagonal matrix maintains robust reconstruction
for larger~$\tau$, which agrees with the theory that non-diagonal matrices generate lattices with longer shortest vectors and hence larger tolerable vector remainder errors.

For the non-diagonal case, the two-stage reconstruction strategy
of Example~\ref{ex:2-stage} is applied using the same setting. The results, also plotted in Fig.~\ref{fig:sim1}, show that the two-stage scheme further enlarges the robust region of the vector remainder error bound. Even when the single-stage system fails, the two-stage approach keeps the reconstruction error small,
demonstrating a clear robustness improvement. The theoretical vector remainder error bounds obtained directly from~(\ref{tau}) and~(\ref{eq:th2-stage}) are $\tau_{\mathrm{diag}} = 43.5$, 
  $\tau_{\mathrm{non\text{-}diag}} = 47.6$, and 
  $\tau_{\mathrm{two\text{-}stage}} = 79.4$, which are marked in Fig.~\ref{fig:sim1}. The reason why the simulation performances for the robust MD-CRTs are slightly better than the theoretical vector remainder error bounds is that these theoretical bounds are only sufficient conditions and behave as lower bounds of the performances. This holds similarly in Fig.~\ref{fig:sim2} below.

Overall, these results confirm that although the single-stage robust MD-CRT with non-diagonal moduli already outperforms the diagonal case, the proposed multi-stage framework can further enhance the vector remainder error bound.

\begin{figure}[t]
  \centering
  \includegraphics[width=\linewidth]{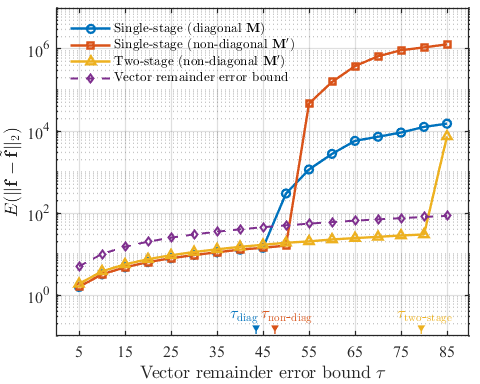}
  \caption{
  Mean reconstruction error versus vector remainder error bound~$\tau$
  for diagonal, non-diagonal, and two-stage systems.}
  \label{fig:sim1}
\end{figure}

We next test the six matrix moduli introduced in the motivating example at the beginning of Section~\ref{s5}. Both single-stage and two-stage reconstructions are performed to further verify the robustness improvement.

The vector remainder error bound is set as $\tau=1,2,\ldots,10$.
For each~$\tau$, $2000$ random trials are conducted.
In every trial, the remainder errors are independently drawn from a uniform distribution within
$\{\mathbf{e}\in\mathbb{Z}^2:\ \|\mathbf{e}\|_2\le\tau\}$,
and the mean reconstruction error $\mathbb{E}[\|\mathbf{f}-\tilde{\mathbf{f}}\|_2]$ is then computed.

As shown in Fig.~\ref{fig:sim2}, the single-stage reconstruction completely fails for all tested $\tau$ values.
In contrast, the proposed two-stage framework
remains stable for $\tau\le7$.
The theoretical vector remainder error bounds $\tau_{\mathrm{single\text{-}stage}} = 0.25$ and $\tau_{\mathrm{two\text{-}stage}} = 6.95$ are also shown in Fig.~\ref{fig:sim2}.
These results clearly demonstrate that
even when a single-stage robust MD-CRT has no robustness,
the multi-stage framework can successfully achieve
robust reconstruction.

\begin{figure}[t]
  \centering
  \includegraphics[width=\linewidth]{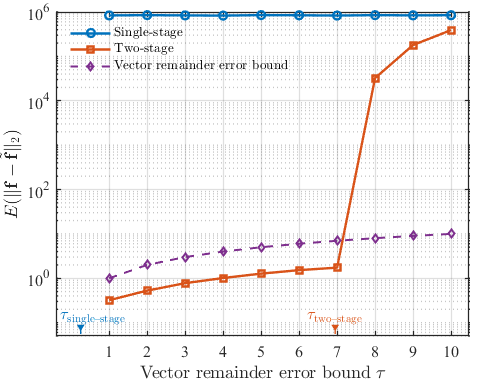}
  \caption{
  Mean reconstruction error versus vector remainder error bound~$\tau$
  for the single-stage and two-stage systems
  using the six matrix moduli in Section~\ref{s5}.}
  \label{fig:sim2}
\end{figure}

\section{Conclusion}\label{s7}
This paper has investigated the MD-CRT for general matrix moduli from the perspectives of dynamic range and robustness.
We showed that under a fixed determinant constraint, i.e., within an upper bound of the sampling rates, non-diagonal matrix moduli do not increase the dynamic range but lead to more balanced and better-conditioned sampling patterns.
They also generate lattices with longer shortest vectors, thus providing stronger robustness to vector remainder errors than that of diagonal matrix moduli, i.e., non-separable robust MD-CRT performs strictly better than separable robust MD-CRT.
To further improve the robustness of the robust MD-CRT, we then further developed a multi-stage robust MD-CRT framework that improves the robustness level without reducing the dynamic range.
Both theoretical analysis and numerical simulations confirmed that the proposed approach achieves higher error tolerance and more reliable reconstruction in multidimensional settings.


\end{document}